\pdfoutput=1 
\documentclass[a4paper,american]{llncs}

\usepackage{microtype}
\usepackage{amssymb,amsmath}
\usepackage{algorithm,algorithmic}
\usepackage[shortlabels]{enumitem}
\usepackage[sort,numbers]{natbib}
\usepackage{pgf}
\usepackage{booktabs}
\usepackage{multicol}
\usepackage{adjustbox}
\usepackage{longtable}
\usepackage{subcaption}
\usepackage{wrapfig}

\newif\ifshowcomments
\showcommentstrue 

\definecolor{orange}{RGB}{255,127,0}

\usepackage{array}
\newcolumntype{H}{>{\setbox0=\hbox\bgroup}c<{\egroup}@{}}

\newcommand{\norm}[1]{\lVert #1 \rVert_2}

\newcommand{\RE}{\mathbb{R}}

\newcommand{\dom}{\mathcal{D}}

\DeclareMathOperator*{\rsp}{rowsp}
\DeclareMathOperator{\dev}{dev}

\newif\ifshowproofs
\showproofsfalse

\let\oldbibliography\bibliography
\renewcommand{\bibliography}[1]{{%
		\let\chapter\section
		\oldbibliography{#1}}}
\begin{document}

\title{Approximate Constrained Lumping}

\institute{
	Aalborg University, Denmark
	\and
    Universidad Politécnica de Madrid, Spain
	\and
	IMT School for Advanced Studies Lucca, Italy
	\and
	Sant'Anna School of Advanced Studies, Pisa, Italy
	\and
	DTU Technical University of Denmark
}

\author{
Alexander Leguizamon-Robayo \inst{1} \and
Antonio Jim\'{e}nez-Pastor \inst{2} \and
Micro Tribastone\inst{3}
Max Tschaikowski \inst{1} \and
Andrea Vandin\inst{4,5}
}

\title{Approximate Constrained Lumping of \\ Chemical Reaction Networks}

\allowdisplaybreaks[0]

\maketitle

\begin{abstract}
Gaining insights from realistic dynamical models of biochemical systems can be challenging given their large number of state variables. 
Model reduction techniques can mitigate this by decreasing complexity by mapping the model onto a lower-dimensional state space. Exact constrained lumping identifies reductions as linear combinations of the original state variables in systems of nonlinear ordinary differential equations, preserving specific user-defined output variables without error.
However, exact reductions can be too stringent in practice, as model parameters are often uncertain or imprecise---a particularly relevant problem for biochemical systems.

We propose approximate constrained lumping. It allows for a relaxation of exactness within a given tolerance parameter $\varepsilon$, while still working in polynomial time.
We prove that the accuracy, i.e., the difference between the output variables in the original and reduced model, is in the order of $\varepsilon$. 
Furthermore, we provide a heuristic algorithm to find the smallest $\varepsilon$ for a given maximum allowable size of the lumped system. 
Our method is applied to several models from the literature, resulting in coarser aggregations than exact lumping while still capturing 
the dynamics of the original system accurately.
\end{abstract}


\section{Introduction}\label{sec_intro}

Biochemical models based on dynamical systems help discover mechanistic principles in living organisms and predict their behavior under unforeseen circumstances.
Realistic and accurate models, however, often require considerable detail that inevitably leads to large state spaces.
This hinders both human intelligibility and numerical/computational analysis.
For example, even the relatively primary mechanism of protein phosphorylation may yield a state space that grows combinatorial with the number of phosphorylation sites~\cite{FEBS:FEBS7027}.

As a general way to cope with large state spaces, model reduction aims to provide a lower-dimensional representation of the system under study that retains some dynamic properties of interest to the modeler.
In applications to computational systems biology, it is beneficial that the reduced state space keeps physical interpretability, mainly when the model is used to validate mechanistic hypotheses~\cite{10.1093/bioinformatics/btn035,21899762,Apri201216}.
There is a variety of approaches in this context, such as those exploiting time-scale separation properties~\cite{okino1998,DBLP:journals/tac/WhitbyCKLTT22}, quasi-steady-state approximation~\cite{SegelS89,radulescu2012reduction}, heuristic fitness functions~\cite{10.1145/3071178.3071265}, spatial regularity~\cite{DBLP:journals/pe/TschaikowskiT17}, sensitivity analysis~\cite{Snowden:2017aa}, and conservation analysis, which detects linear combinations of variables that are constant at all times~\cite{10.1093/bioinformatics/bti800}.

By \emph{lumping} one generally refers to a reduction method that provides a self-consistent system of dynamical equations comprised of a set of macro-variables, each given as a combination of the original ones~\cite{aokiOptimizationStochasticSystems1967, kemenyFiniteMarkovChains1960,okino1998,Snowden:2017aa}.
In linear lumping, this reduction is expressed as a linear transformation of the original state variables.
Since this can destroy physical intelligibility in general, \emph{constrained lumping} allows restricting to only part of the state space by defining linear combinations of state variables that ought to be preserved in the reduction~\cite{LI199195}.
Lumping techniques of Markov chains date back to the early nineties~\cite{Larsen19911}.
These were later extended to stochastic process algebra~\cite{10.1007/978-0-387-09680-3_18,10.1007/978-3-540-88479-8_13,10.1093/comjnl/bxr094,DBLP:conf/splc/Tribastone14,TSCHAIKOWSKI2014140}, and have been expanded more recently to efficient algorithmic approaches for stochastic chemical reaction networks~\cite{DBLP:journals/bioinformatics/CardelliPTTVW21} and deterministic models of biological systems based on ordinary differential equations (ODEs) with polynomial right-hand sides~\cite{CARDELLI2019132,pnas17,DBLP:conf/lics/CardelliTTV16}.

In these cases, the reduction is a specific kind of linear mapping induced by a partition (i.e., an equivalence relation) of state variables; in the aggregated system, each macro-variable represents the sum of the original variables of a partition block. 
For polynomial ODE systems, CLUE has been presented as an algorithm that computes constrained linear lumping efficiently, i.e., in polynomial time~\cite{ovchinnikov_clue_2021}, by avoiding the symbolic computation of the eigenvalues of a non-constant matrix that would have required in prior approaches~\cite{LiRabitz,LI199195}.

In particular, CLUE can compute \emph{the smallest} linear dimensional reduction that preserves the dynamics of arbitrary linear combinations of original state variables given by the user. 
A complete step-by-step example of CLUE reduction on a synthetic 3-variables model is provided later (Exampe~\ref{ex:firstode}). 
To give a concrete example of the usefulness of CLUE applied to a real model, 
we consider a model of FceRI-like network of a cell-surface receptor~\cite{borisov_domain-oriented_2008}.
The \emph{observable} or quantity of interest is the amount of \emph{Phosphorylation at site Y2 (S2P)}, which
is the sum of 10 out of the 24 variables in the model.
Figure~\ref{fig:introPlot:exact}, displays the original simulation in blue and the orange line, instead, shows the solution of such observable on a model containing 19 variables obtained by reducing the former by CLUE. 
This plot shows that one can study the evolution of \emph{S2P} using either the original model or one reduced by constrained lumping with only 19 variables without introducing any error.
\begin{figure}[H]
    \centering
    \begin{subfigure}[t]{0.45\textwidth}
		\centering
    \includegraphics[width=\linewidth]{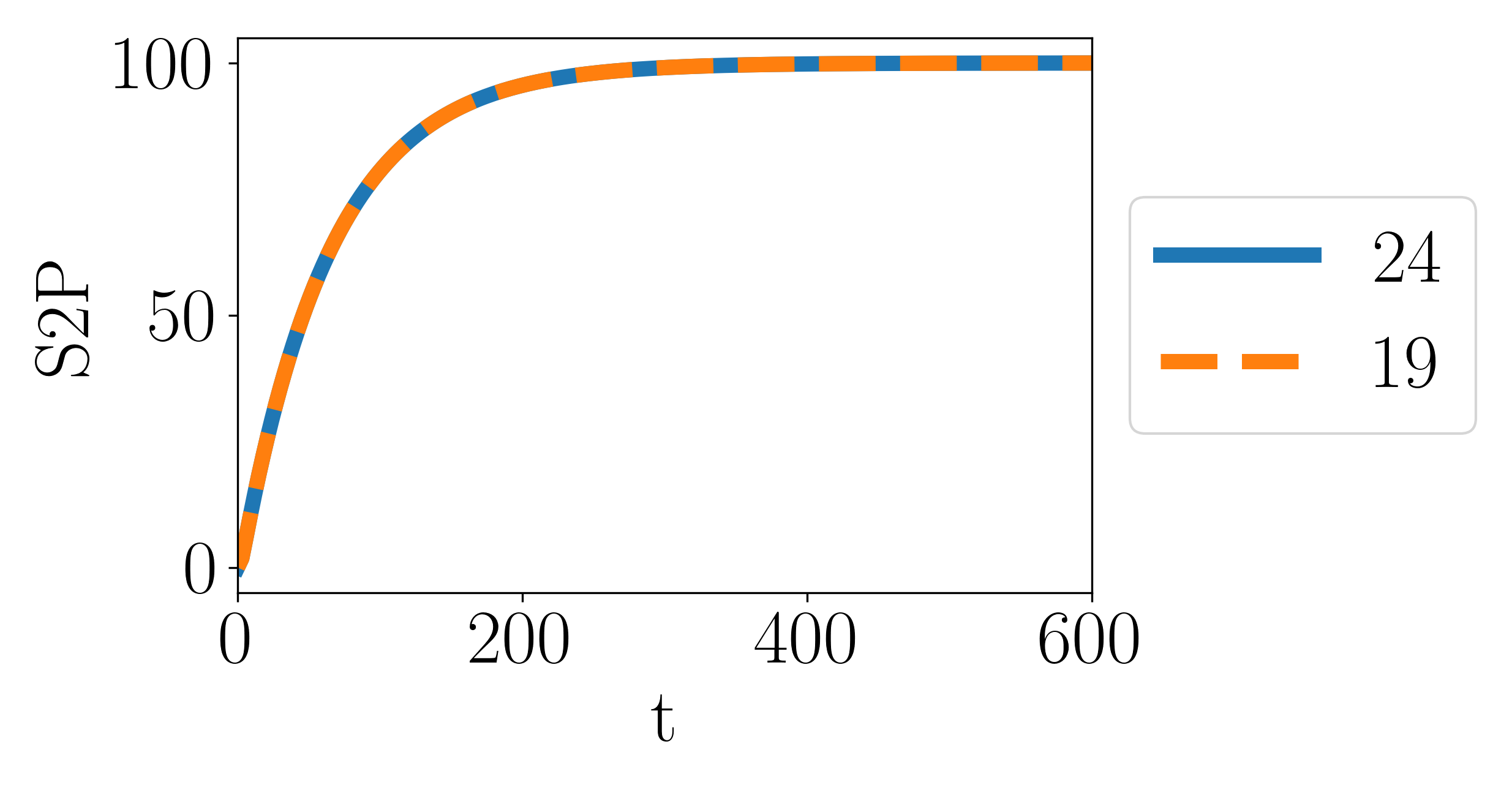}
		\caption{ Exactly reduced model by CLUE with 19 variables (orange).}\label{fig:introPlot:exact}
	\end{subfigure}
	\hspace{1cm}
 \begin{subfigure}[t]{0.45\textwidth}
		\centering
    \includegraphics[width=\linewidth]{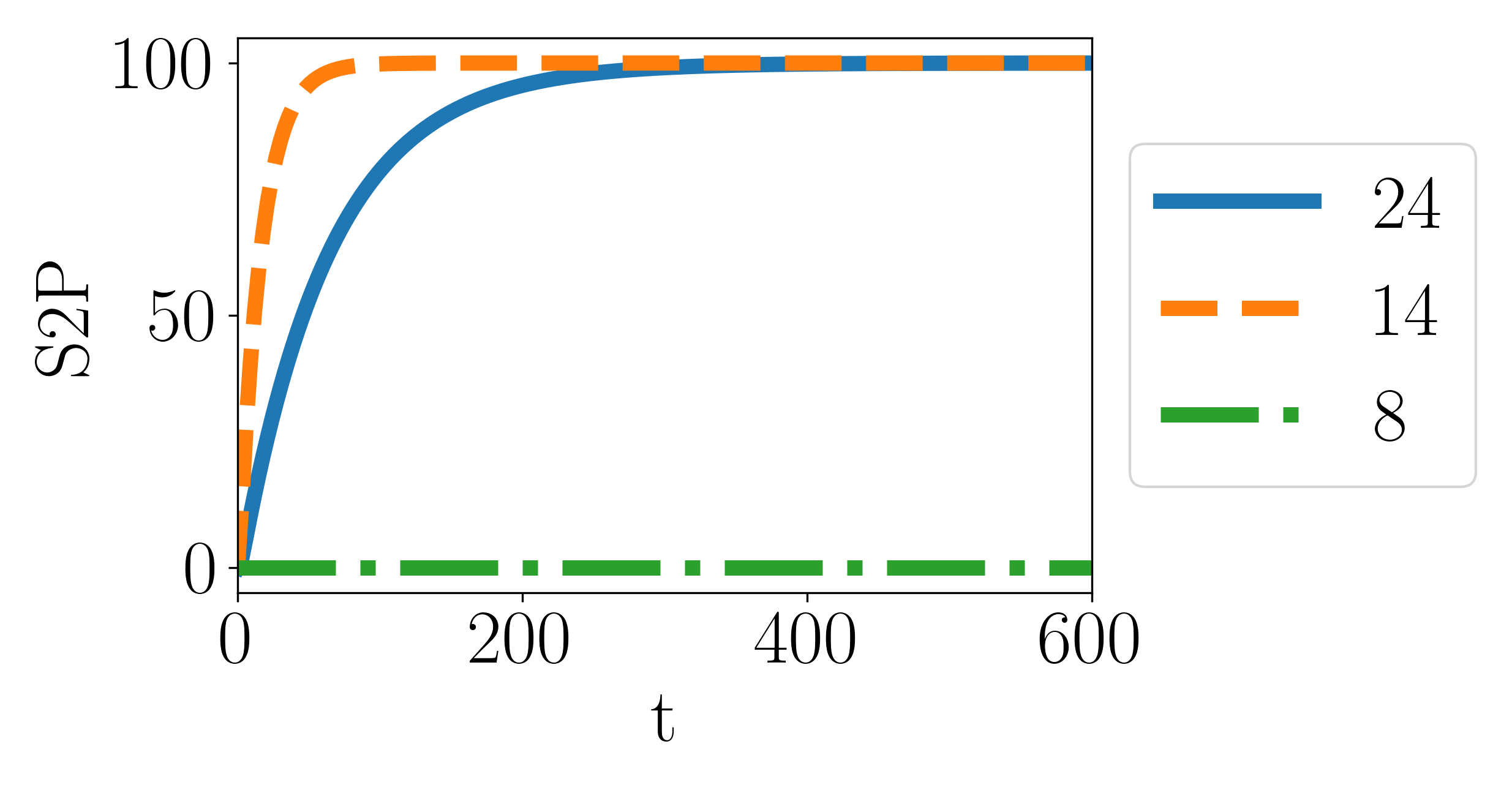}
		\caption{Approximately reduced model with 14 (orange) and 8 (green) variables.}
  \label{fig:introPlot:app}
	\end{subfigure}
	\caption{Evolution of observable S2P in a FceRI-like network of a cell-surface receptor~\cite{borisov_domain-oriented_2008} using different models.
    The simulation using the original model with 24 variables is displayed in blue.
 }  
  \label{fig:introPlot}
\end{figure}

All the aforementioned lumping approaches are  \emph{exact}, i.e., the reduced model does not incur any approximation error (but only loss of information because, in general, the aggregation map is not invertible).
Approximate lumping is a natural extension that has been studied for a long time (e.g.,~\cite{LI1990977}).
Indeed, although exact reduction methods have been experimentally proved successful in a large variety of biological systems (e.g.,~\cite{SBML}), approximate reductions can be more robust to parametric uncertainty---which notoriously affects systems biology models (e.g.,~\cite{doi:10.1098/rsif.2017.0237,DBLP:conf/cmsb/BarnatBBDHPS17})---and offer a flexible trade-off between the aggressiveness of the reduction and its precision.
This has been explored for partition-based lumping algorithms:
In~\cite{DBLP:conf/qest/CardelliTTV18,cardelliFormalLumpingPolynomial2023}, the authors present an algorithm for approximate aggregation parameterized by a tolerance $\varepsilon$, which, informally, relaxes an underlying criterion of equality for two variables to be exactly lumped into the same partition block. 
These approaches present notions of approximate reduction extending the above-mentioned exact reduction techniques for chemical reaction networks and ODEs. 
Similarly, in this paper, we present an extension of CLUE~\cite{ovchinnikov_clue_2021} that performs approximate constrained linear lumping of ODE systems. 
It relaxes the conditions for (exact) CLUE by introducing a \emph{lumping tolerance} parameter that roughly relates to how close a lumping matrix is to an exact lumping. 
We consider ODE systems with analytical derivatives.
These include \textit{rational} and \textit{polynomial} derivatives thus covering kinetic models for biochemical systems with Hill~\cite{Voit:2013aa}, Michaelis-Menten~\cite{Voit:2013aa} and mass-action kinetics (e.g.,~\cite{Voit:2013aa,DBLP:journals/biosystems/TrojakSPB23}). 

To obtain an approximate reduction for a given analytical ODE system, we start by finding a finite representation of the dynamics, either symbolically or using automatic differentiation~\cite{jimenez_clue_2022}.
Given a user-defined lumping tolerance, we can compute a reduced system that approximates the original one up to an error proportional to the lumping tolerance. 
For example, the orange line in Figure~\ref{fig:introPlot:app} shows that we can study the evolution of $S2P$ using an approximate reduction with 14 variables that yield the same steady-state value as in the original model. However, a too-permissive lumping tolerance may destroy the original dynamics, as shown by the approximate reduction with 8 variables (bottom green line).

Our proposed approach works in polynomial time. 
To find the lumping tolerance, we propose a heuristic approach based on the expected size of the reduced model.
%
Using a prototype implementation, we evaluate the aggregation power of our approach on five polynomial models and three rational models representative of the literature.
We also evaluate the scalability of our approach on a multisite phosphorylation model~\cite{sneddon2011efficient}.
Overall, the numerical results show that our method can lead to substantially smaller reduced models while introducing limited errors in the dynamics.

This paper extends the conference paper~\cite{leguizamon-robayo_approximate_2023} by expanding the theory to include analytical drifts; including models with rational drifts; adding a scalability analysis; and proposing a new approach to find the lumping tolerance.
A short tool demonstration paper discussing implementation details of our approach has been recently presented in~\cite{cmsb2024_clue}.

\smallskip

\noindent\emph{Further related work.} 
We are complementary to classic works~\cite{LI1990977,LI199195} which do not address the efficient algorithmic computation of lumping. Moreover, we are more general than~\cite{iacobelli_lumpability_2013,tac15,DBLP:conf/qest/CardelliTTV18} which consider so-called ``proper'' lumping (i.e.,~\cite{Snowden:2017aa}) where each state variable appears in exactly one aggregated variable. 
Likewise, we are complementary to rule-based reduction techniques~\cite{Feret21042009} which are independent of kinetic parameters~\cite{concur15}. 
As less closely related abstraction techniques, we mention (bisimulation) distance approaches for the approximate reduction of Markov chains~\cite{DBLP:conf/tacas/BacciBLM13,DBLP:conf/concur/DacaHKP16} and proper orthogonal projection~\cite{antoulas}, which, however, apply to linear systems. Abstraction of chemical reaction networks through learning~\cite{DBLP:journals/iandc/RepinP21,DBLP:conf/cmsb/CairoliCB21} and simulation~\cite{DBLP:conf/cmsb/HelfrichCKM22} are complementary to lumping. 

\smallskip

\noindent\emph{Outline.} 
Section~\ref{sec_pre} provides necessary preliminary notions.
Section~\ref{sec_theory} introduces approximate constrained lumping, how to compute it, and how to choose the lumping tolerance value.
Section~\ref{sec_eval} evaluates our proposal on models from the literature, while Section~\ref{sec_conc} concludes the paper.

\smallskip

\noindent\emph{Notation.} 
For a function $f$, we denote its domain by $\dom(f)$.
The derivative with respect to time of a function $x : [0,T] \to \RE^m$ is denoted by $\dot{x}$.
We denote a dynamical system by $\dot{x} = f(x)$, and denote its initial condition by $x(0)=x^{0}$.
For $f: \mathbb{R}^{m} \to \mathbb{R}^{n}$,  we denote the Jacobian of $f$ at $x$ by $J(x)$.
Given a matrix $L \in \mathbb{R}^{m\times n} $, the rowspace of $L$ or the vector space generated by the rows of $L$ is denoted by $\rsp(L)$.
We denote by $\bar{L} \in \mathbb{R}^{n \times m}$ a  pseudoinverse of $L$ (i.e. $L\bar{L}= I_{n}$ where $I_{n} \in \mathbb{R}^{n\times n}$ is the identity matrix).
We reserve the term \textit{rational} for models with a non-trivial denominator on their right-hand side.


\section{Preliminaries}\label{sec_pre}

In this paper, we study systems of ODEs of the form:
\begin{equation}
 \dot{x}  = f(x), \label{eq:ode}
\end{equation}
where $f: \mathbb{R}^{m} \to \mathbb{R}^{m}, x\mapsto (f_{1}(x), \dots, f_{m}(x))^{T}$ and $f_{i}$ is an analytic function,  for $i=1,\dots, m$.

We first recall the notion of (exact) lumping~\cite{tomlin_effect_1997}. 
Informally, it is defined as a mapping $L$ that leads to a self-consistent system of ODEs in the reduced state space.
\begin{definition}\label{def:lump}
 Given a system of ODEs of the form given by Equation~\eqref{eq:ode}, and a full rank matrix $L\in \mathbb{R}^{l\times m}$ with $l< m$,
 we say that $L$ is an \emph{exact lumping of dimension $l$} (or that the system is \emph{exactly lumpable} by $L$) if there exists a function $g: \mathbb{R}^{l} \to \mathbb{R}^{l}$ with polynomial entries such that $Lf = g\circ L$.
\end{definition}

\noindent \begin{example}
 \label{ex:firstode}
 Consider the following system
 \begin{equation} \label{eq:example1}
  \dot{x}_{1} = \frac {x_{2}^{2} +4x_{2}x_{3} +4x_{3}^{2}}{x_{1}^{2} + 1 },\qquad
  \dot{x}_{2} = \frac{2x_{1}-4x_{3}}{x_{2}+2x_{3} + 1},\qquad
  \dot{x}_{3} = \frac{-x_{1}-x_{2}}{x_{2}+2x_{3}+1}.
 \end{equation}
 Then, the matrix 
 \begin{equation*}
     L = \begin{pmatrix}
         1 & 0 & 0\\
         0 & 1 & 2
     \end{pmatrix}
 \end{equation*}
  is an exact lumping of dimension 2, since
 \begin{equation*}
  \begin{pmatrix}
   \dot{y}_{1} \\
   \dot{y}_{2}
  \end{pmatrix}
  =
  \begin{pmatrix}
   \dot{x}_{1} \\
   \dot{x}_{2} +2 \dot{x}_{3}
  \end{pmatrix}
  =
  \begin{pmatrix}
   \frac{(x_{2}+2x_{3})^{2}}{x_{1}^{2}+1} \\
   \frac{-2x_{2}-4x_{3}}{x_{2}+2x_{3}}
  \end{pmatrix}
  =
  \begin{pmatrix}
   \frac{y_{2}^{2}}{y_{1}^{2}+1} \\
   \frac{-2y_{2}}{y_{2}+1}
  \end{pmatrix}
  =
  \begin{pmatrix}
   g_{1}(y_{1},y_{2}) \\
   g_{2}(y_{1},y_{2})
  \end{pmatrix}.
 \end{equation*}
\end{example}

\begin{definition}\label{def:constlump}
 Given a system of ODEs of the form given by Equation~\eqref{eq:ode}, and an exact lumping $L\in \mathbb{R}^{l\times m}$, we say that $y= Lx$ are the \emph{reduced (or lumped) variables}, 	and their evolution is given by the \emph{reduced system} $ \dot{y}= g(y)$.
\end{definition}

Given an initial condition $x^{0} \in \mathbb{R}^n$ and an exact lumping $L$, there is a corresponding initial condition $y^{0}$ in the lumped variables given by $y^{0} = L x^{0}$.
Similarly, since $y=Lx$ and $g\circ L = Lf$, we have that \[L\dot{x} = Lf(x)= g(Lx)=g(y)=\dot{y}\ .\]
This means that we can study the evolution of the lumped variables $y(t)$ by solving the (smaller) reduced system rather than the (larger) original one.

To answer whether it is possible to recover the evolution of some linear combination of state variables., we introduce the notion of constrained lumping.

\begin{definition}
 Let $x_{obs}= Mx$ for some matrix $M \in \mathbb{R}^{p\times m}$, for $p< m$.
 We say that a lumping $L$  is a \emph{constrained lumping} with \emph{observables} $x_{obs}$ if  $\rsp(M) \subseteq \rsp(L)$.
 This means that each entry of $x_{obs}$ is a linear combination of the reduced variables $y$. 	
\end{definition}

\begin{example}
 Suppose we are interested in observing the quantity $2x_{1}+x_{2}+2x_{3}$ where the evolution is given by the system~\eqref{eq:example1}.
 In this case $x_{obs}= Mx$ with $M = ( 2\ 1\ 2)$.
 We can see that  $L$ from Example \ref{ex:firstode} is a constrained lumping as we can recover the observable from the reduced system, i.e., $ x_{obs}=2y_{1}+y_{2}$.
 Suppose now that we want to observe the quantity $x_{1}+x_{2}+x_{3}$.
 In this case $M = ( 1\ 1\ 1)$, thus matrix $L$, is not a constrained lumping as it is not possible to obtain $x_{obs}$ as a linear combination of $y_{1}$ and $y_{2}$.
\end{example}

To understand how a constrained lumping can be computed, we first review the following known characterization of lumping.

\begin{theorem}[Characterization of Exact Lumping~\cite{tomlin_effect_1997}]\label{thm:lumping}
 Given a system of $m$ ODEs of the form given by Equation~\eqref{eq:ode} and a matrix $L\in \mathbb{R}^{l\times m}$ with rank $l$, the following are equivalent.
 \begin{enumerate}
  \item The system is exactly lumpable by $L$.
        \item\label{thm:lumping:inverse} For any pseudoinverse $\bar{L}$ of $L$, $Lf = (Lf) \circ \bar{L}L$.
        \item\label{thm:lumping:invariant} $\rsp(L J(x)) \subseteq \rsp(L)$  for all $x \in \RE^m$; that is, the row space of $L$ is invariant under $J(x)$ for all $x\in \mathbb{R}^{m}$, where $J(x)$ is the Jacobian of $f$ at $x$.
 \end{enumerate}
\end{theorem}

The characterization of exact lumping in Theorem~\ref{thm:lumping} provides us with a way to compute a constrained lumping $L$.
This is because, thanks to point~\ref{thm:lumping:invariant}, the problem of computing a lumping is equivalent to the problem of finding a $J(x)$-invariant subspace of $\mathbb{R}^{m}$ for all $x \in \mathbb{R}^{m}$.
However, $J(x)$ is a matrix whose entries are functions of $x$.
In other words, there is a different real-valued matrix for each $x \in \mathbb{R}^{m}$.
To circumvent this issue, we require the following result:

\begin{theorem}[{\cite[Lemma 1]{jimenez_clue_2022}}]\label{thm:repJ}
 Consider a system of ODEs of the form given by Equation~\eqref{eq:ode} and let $J(x)$ be the Jacobian matrix of $f$.
 Let $\mathcal{B}$ be any set of real-valued matrices spanning the $\mathbb{R}-$vector space $$\mathcal{V}_{J}:=\left< J(x) | x \in \mathbb{R}^{m} \text{ and $J(x)$ is well-defined}\right>.$$
 Then $L$ is a lumping of~\eqref{eq:ode} if and only if $\rsp(L)$ is invariant to each $J_i \in \mathcal{B}$.
\end{theorem}

A set of matrices $\{J_1,\dots, J_N\}$ can be found analytically when $J(x)$ can be represented as
\begin{equation}\label{eq:jacrep}
 J(x) = \sum_{i=0}^{N}J_{i}\mu_{i}(x),
\end{equation}
where $\left\{\mu_{i}(x) \mid  0 \leq i \leq N \right\}$ is a set of analytic functions.
When working with polynomial systems, each of the entries of $J(x)$ is a polynomial and so each $\mu_i$ corresponds to the monomials of $J(x)$. 

For instance, for rational systems, $J(x)$ can be obtained symbolically by differentiating $f(x)$ and then multiplying by the minimum common denominator (mcd) of all entries of $J(x)$. 
In practice, this approach is computationally unfeasible since this symbolic approach results in an explosion of the number of monomials in $J(x)$.
To avoid any explicit symbolic computation, a representation $\{J_1,\dots, J_N\}$ can be obtained by sampling $f(x)$ at different points and using automatic differentiation~\cite[Section 3.4]{jimenez_clue_2022}. 
The general idea of this sampling procedure is formalized by Algorithm~\ref{alg:findJ}. 
In~\cite{jimenez_clue_2022}, the authors argue that a uniform sampling of $J(x)$ values finds almost surely $\mathcal{V}_{J}$.

 \begin{algorithm}[t]
  \caption{Computation of $\mathcal{V}_{J}$ from Theorem~\ref{thm:repJ}.}\label{alg:findJ}
  \begin{algorithmic}[1]
   \REQUIRE
   $f: \mathbb{R}^{m} \to \mathbb{R}^{m}$, and 
   a compact set $\Omega \subseteq \mathbb{R}^{m}$\\
   \STATE {\textbf{set} $N=0$}
   \REPEAT
   \STATE {\textbf{set} $N=N+1$}
   \STATE {\textbf{compute} $x_N$ sampling uniformly from $\Omega$ }
   \STATE {\textbf{compute} $J_N=J(x_N)$ using automatic differentiation}
   \UNTIL{$J_N \in \left< J_i | 1\leq i \leq N-1 \right>$}

   \RETURN {$\{J_1,\dots, J_{N-1}\}$}.
  \end{algorithmic}
 \end{algorithm}

  \begin{example}\label{ex:findJ} 
  Consider the system from  Example~\ref{ex:firstode}.
  Before finding a basis for $\mathcal{V}_{J}$ using the sample-based approach, we exemplify the symbolic approach to display its complexity.  
   It starts by 
   symbolically differentiating $f(x)$ to obtain $J(x)$:
   \begin{equation*}
J(x) = \begin{pmatrix}
\frac{-2x_1x_2^2 - 8x_1x_2x_3 - 8x_1x_3^2}{x_1^4 + 2x_1^2 + 1} & \frac{2x_2 + 4x_3}{x_1^2 + 1} & \frac{4x_2 + 8x_3}{x_1^2 + 1}\\
\frac{2}{x_2 + 2x_3 + 1} & \frac{-2x_1 + 4x_3}{x_2^2 + 4x_2x_3 + 4x_3^2 + 2x_2 + 4x_3 + 1} & \frac{-4x_1 - 4x_2 - 4}{x_2^2 + 4x_2x_3 + 4x_3^2 + 2x_2 + 4x_3 + 1}\\
\frac{-1}{x_2 + 2x_3 + 1} & \frac{x_1 - 2x_3 - 1}{x_2^2 + 4x_2x_3 + 4x_3^2 + 2x_2 + 4x_3 + 1} &\frac{2x_1 + 2x_2}{x_2^2 + 4x_2x_3 + 4x_3^2 + 2x_2 + 4x_3 + 1}
\end{pmatrix}.
   \end{equation*}
   We then find the mcd of all entries of $J(x)$, which is 
   $ q(x)= (x_{1}^{2} +1)^{2}(x_{2}+2x_{3}+1)^{2}$.
   %
   Since $J(x)$ is rational in all its entries, it follows that $J(x)= B(x)/q(x)$, where 
   $B(x) = q(x)J(x)$ is a matrix with polynomial entries up to degree 5.
   Notice that to obtain a monomial representation of $B(x)$, we need to expand all its terms.
   %

   These complex symbolic computations can be avoided by using Algorithm~\ref{alg:findJ}.
   Following the efficient sampling method outlined in~\cite{jimenez_clue_2022}, we find that the dimension of $\mathcal{V}_{J}$ is 5.
   The set $\{J_1,\dots,J_5\}$ can be obtained by evaluating $J(x)$ at the following points: $
    x_1 = (1, 0, 0),~x_2 = (0, 1, 0),~x_3 = (0, 0, 1),~x_4 = (1, 5, 2),~x_5 = (3, 3, 2).$
   These points lead to the matrices $J_i$: 
   \begin{align*}
     J_1 &=
     \begin{pmatrix}
      0  & 0  & 0  \\
      2  & -2 & -8 \\
      -1 & 0  & 2
     \end{pmatrix},&
     J_2 &=
     \begin{pmatrix}
      0    & 2     & 4   \\
      1    & 0     & -2  \\
      -0.5 & -0.25 & 0.5
     \end{pmatrix},\\
     J_3 &=
     \begin{pmatrix}
      0      & 4      & 8      \\
      0.667  & 0.444  & -0.444 \\
      -0.333 & -0.333 & 0
     \end{pmatrix},&
     J_4 &=
     \begin{pmatrix}
      -40.5  & 9      & 18     \\
      0.200  & 0.060  & -0.280 \\
      -0.100 & -0.040 & 0.120
     \end{pmatrix},
     \\
     J_5 &= 
     \begin{pmatrix}
         -2.94 & 1.4 & 2.8 \\
         0.25 & 0.31 & -0.438\\
         -0.125 & -0.031 & 0.188
     \end{pmatrix}.
   \end{align*}
  \end{example}

 We see that Algorithm~\ref{alg:findJ} outputs a set of matrices $\{ J_1,\dots, J_N\}$
 spanning $\mathcal{V}_{J}$.
 Given such a set, Theorem~\ref{thm:repJ} provides an algorithmic way to find a constrained lumping $L$ by performing a finite number of checks over real-valued vectors.
 This is shown in Algorithm~\ref{alg:findL}.
 The use of Algorithms~\ref{alg:findJ} and~\ref{alg:findL} to obtain a constrained lumping $L$ is summarized in Algorithm~\ref{alg:clue}.

 \hspace{-0.65cm}
 \begin{minipage}[t]{0.47\textwidth}
  \begin{algorithm}[H]
   \caption{Computation of $L$}\label{alg:findL}
   \begin{algorithmic}[1]
    \REQUIRE
    a set of matrices $\left\{ J_1,\dots,J_N \right\}$ spanning $\mathcal{V}_{J}$ (Theorem \ref{thm:repJ});\\
    a $p \times m$ matrix $M$ with row rank $p$.\label{alg:findL:M}

    \STATE \textbf{set} $L := M$
    \REPEAT
    \FORALL{$1 \leq i \leq \kappa$ and rows $r$ of $L$}
    \STATE {\textbf{compute} $r J_i$ }
    \IF{$rJ_{i}\notin \rsp(L)$} \label{alg:findL:check}
    \STATE {\textbf{append} row $r J_i$ to $L$}\label{alg:findL:append}
    \ENDIF
    \ENDFOR
    \UNTIL{no rows are appended to $L$}

    \RETURN Lumping matrix $L$.
   \end{algorithmic}
  \end{algorithm}
  \vspace{0.01cm}
 \end{minipage}
 \hfill
 \begin{minipage}[t]{0.47\textwidth}
  \begin{algorithm}[H]
   \caption{Constrained lumping \cite{ovchinnikov_clue_2021} }\label{alg:clue}
   \begin{algorithmic}[1]
    \REQUIRE \!\!a\! system \!$\dot{x} \!=\!\! f(x)\!$ of $\!m\!$ \! ODEs;\\
    a $p \times m$ matrix $M$ with row rank $p$.\\


    \STATE \textbf{compute} a set of matrices $\left\{ J_1,\dots,J_N \right\}$ spanning $\mathcal{V}_{J}$ (Algorithm~\ref{alg:findJ})\\
    \STATE \textbf{compute} $L$ (Algorithm \ref{alg:findL}) \label{alg:clue:findL}

    \RETURN Constrained lumped ODE system $\dot{y} = L f (\bar{L}y)$.

   \end{algorithmic}
  \end{algorithm}
  \vspace{0.01cm}
 \end{minipage}

 The number of checks in Algorithm~\ref{alg:findL} is proportional to the number of elements in $\mathcal{V}_{J}$.
 Moreover, Algorithm~\ref{alg:findL} can be implemented in polynomial time~\cite{ovchinnikov_clue_2021,jimenez_clue_2022}.

\section{Constrained Approximated Lumping}\label{sec_theory}

This section introduces the formal theory underlying approximate constrained lumping, as well as its computational aspects.
In Section~\ref{sec:theoerrest} we present the definition of approximate constrained lumping. 
We introduce the \emph{deviation tolerance} $\eta$, which relaxes the conditions for lumping granting additional aggregation power. 
We prove how the error bound depends on the deviation tolerance.
Next, we focus on how to compute an approximate constrained lumping.
In Section~\ref{sec:algacl} we introduce a numerical lumping tolerance $\varepsilon$ to compute a lumping matrix $L$ and show how it relates to the deviation tolerance. 
In Section~\ref{sec:besteps}, instead, we discuss how to pick an appropriate value for the lumping tolerance $\varepsilon$ in a way that leads to satisfactory results for both polynomial and rational systems.

\subsection{Definition and error estimation}\label{sec:theoerrest}
We begin by demonstrating why the condition for exact lumping can be too restrictive for finding reductions.
Consider the following system of ODEs:

\begin{equation} \label{eq:examplepert}
	\dot{x}_{1} = \frac {x_{2}^{2} +4.05x_{2}x_{3} +4x_{3}^{2}}{x_{1}^{2} + 1 },~
	\dot{x}_{2} = \frac{2x_{1}-4x_{3}}{x_{2}+2x_{3} + 1},~
	\dot{x}_{3} = \frac{-x_{1}-x_{2}}{x_{2}+2x_{3}+1}.
\end{equation}

Notice that the system given by Equation~\eqref{eq:examplepert} is similar to that of Equation~\eqref{eq:example1}, as we have just added a term $0.05x_2x_3$ to the numerator of its first equation. While the initial system is lumpable by $L = (1 \ 0\ 0, 0\ 1\  2)^{T}$, this new system is not exactly lumpable by  $L$.
However, we would like to know if it is still possible to use the matrix $L$ to obtain a \textit{useful} reduced system that approximates well the original one.
To do so, we first want to evaluate how \emph{close} the matrix $L$ is to (i.e., how much it \emph{deviates} from) an exact lumping.

By Theorem~\ref{thm:lumping}, Point \ref{thm:lumping:inverse}, a full rank matrix $L$ is a lumping if and only if the equality $Lf(\bar{L}Lx) - Lf(x) = 0$ is satisfied for all $x \in \dom (f) \subseteq \mathbb{R}^{m}$, where $\dom(f)$ denotes the domain of $f$.
In our proposal of approximate lumping, we relax this requirement by asking it to be satisfied up to a certain tolerance.

	\begin{definition}\label{def:deviation}
		Consider the system of ODEs in Equation~\eqref{eq:ode}.
		Let $L \in \mathbb{R}^{l\times m}$ be a full rank  matrix with $l< m$
		and denote by $\bar{L}$ the Moore-Penrose pseudoinverse of $L$ and
		by $P_L= \bar{L}L$  the orthogonal projection operator onto $\rsp(L)$.
		We define the \emph{deviation of $L$ at $x \in \dom(f) \cap \dom(f\circ P_L)$} by
		\begin{equation}
			\dev_{L}(f,x):= \norm{ L f(\bar{L}L x) - L f(x)} .
		\end{equation}
	\end{definition}
\begin{figure}
    \centering
    \includegraphics[width=0.4\linewidth]{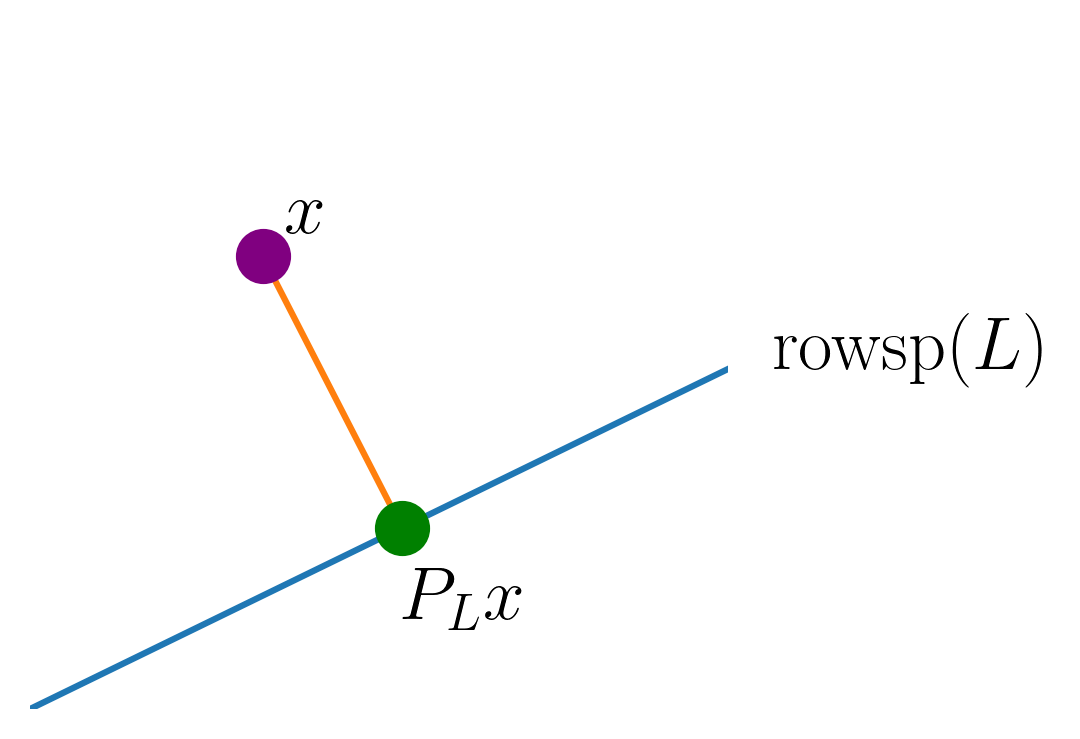}
	\caption{Points projected to $\rsp(L)$~\cite{leguizamon-robayo_approximate_2023}.}
	\label{fig:devexp}
\end{figure}


To understand the intuition behind Definition \ref{def:deviation}, consider the point $x$ in Figure~\ref{fig:devexp}.
Then, the deviation computes the difference between the images under $L f$ of  $x$ (purple) and $P_L x$ (green).
The deviation is identically zero if and only if $L$ is a lumping.

\begin{example} \label{ex:dev}
	Consider the matrix $L$ given in Example \ref{ex:firstode}.
	A pseudoinverse	of $L$ is given by $ \bar{L}=(1\ 0\ 0,0\ 0.2\ 0.4)^{T}$.
	Writing $g$ for the corresponding vector field, we obtain that $dev_{L}(g,(1,1,1)^{T}) = 0$, as $L$ is an exact lumping.
	Now consider the system given by Equation~\eqref{eq:examplepert} and use $h$ to denote the corresponding vector field.
	Then, we get that $\dev_{L}(h,(1,1,1)^{T}) = 0.007$.
	Therefore, the matrix $L$ is not an exact lumping of Equation \eqref{eq:examplepert}.
\end{example}

From a modelling perspective, a reduced model is meaningful insofar as its predictions for a set of initial conditions of interest are \textit{close enough} to those of the original model throughout a given finite time horizon.
While the theory in~\cite{ovchinnikov_clue_2021,jimenez_clue_2022} guarantees that reduced models provide exact predictions, this might restrict the actual aggregation power of the technique.
We aim to relax this theory by considering reductions that need not be exact or tight.
Having this in mind, we introduce the following notion.

	\begin{definition}
		Consider the ODEs from Equation \eqref{eq:ode}, a set of initial conditions $S$, and a finite time horizon $T>0$ such that $x$ is well-defined on $[0,T]$ for all initial conditions $x(0)\in S$.
		Given a full rank  matrix $L \in \mathbb{R}^{l\times m}$ and $\eta >0$, we say that Equation \eqref{eq:ode} is \emph{approximately lumpable} by $L$ with \emph{deviation tolerance} $\eta$ if
		\begin{equation}
			\dev_{L}(f,x(t)) \leq \eta,
		\end{equation}
		for all $t \in [0,T]$.
		We say that $L$ is an \emph{approximate lumping} for the set $S$, time horizon $T$, and with deviation tolerance $\eta$ (or $(S,T,\eta)$-lumping).
		We will omit $S$, $T$ or $\eta$ whenever they can be inferred from the context.
	\end{definition}

\begin{remark}
	The notion of approximate lumping generalizes that of exact lumping from Definition~\ref{def:lump}.
	To see this, suppose $L$ is an exact lumping.
	Then by Point \ref{thm:lumping:inverse} of Theorem \ref{thm:lumping}, $\norm{ L f(\bar{L}L x) - L f(x)} = \dev_{L}(f,x) = 0 $, for all $x \in \mathbb{R}^{m}$.
	It follows that $L$ is a $(\mathbb{R}^m,\infty,0)$-lumping.
\end{remark}

	In practice, modellers often use biological models to study the evolution of a system $\dot{x}= f(x)$ for a fixed set $S$ of initial conditions and a 
	time horizon $T$. The 
	assumption is that the behaviour of interest to the modeler occurs within time $T$.

 \begin{figure}
     \centering
	\includegraphics[width=0.4\linewidth]{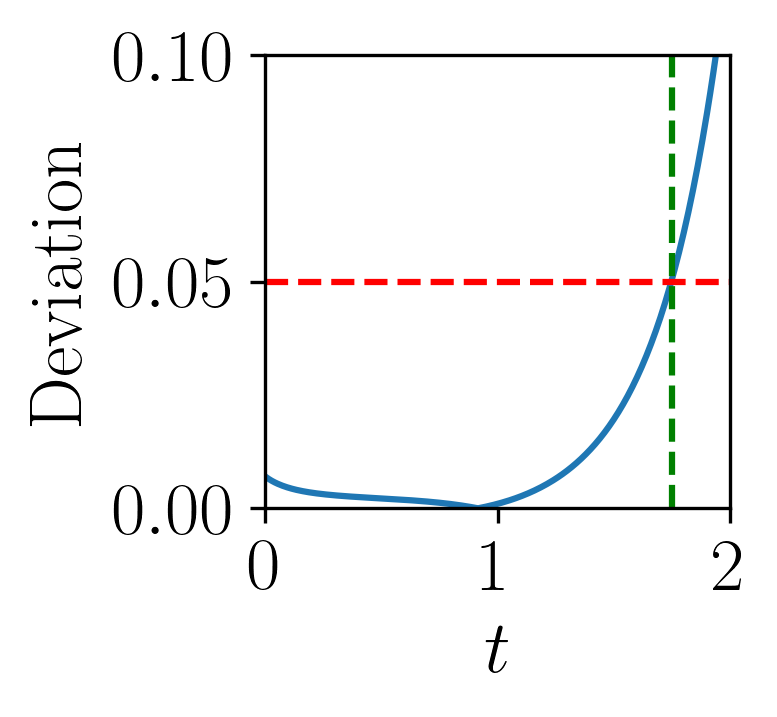}
	\caption{Example \ref{ex:applump}: evolution of $\dev_{L}(f(x(t)))$.}
	\label{fig:ex:lumpscheme}
 \end{figure}

\begin{example} \label{ex:applump}
	Consider the system in Equation \eqref{eq:examplepert}, the matrix $L$ of Example \ref{ex:firstode} and let $x(0)= (1,1,1)^{T}$.
	Then $L$ is an approximate lumping for $x(0)$, time $T=1.75$, and deviation tolerance $0.05$, i.e., $L$ is a $(\{x(0)\}, 1.75, 0.05)$-lumping.
	To see this, note that in Figure \ref{fig:ex:lumpscheme} the deviation of the dynamics is bounded by $0.05$.
\end{example}

After having generalized the notion of lumping to approximate lumping, we next introduce approximate constrained lumping in an obvious manner.

\begin{definition} \label{def:constapplump}
	Let $x_{obs}= Mx$ for some matrix $M \in \mathbb{R}^{p\times m}$ with $p < m$.
	We say that an approximate lumping $L$ of Equation \eqref{eq:ode} is an \emph{approximate constrained lumping} with \emph{observables} $x_{obs}$ if $\rsp(M) \subseteq \rsp(L)$.
\end{definition}

Similarly to Definition \ref{def:constlump}, Definition \ref{def:constapplump} means that the observables $x_{obs}$ can be recovered as a linear combination of the reduced variables $y = Lx$.

\begin{definition}\label{def:redsys}
	The \emph{reduced system} induced by an approximate lumping $L$ is given by $\dot{y} = Lf(\bar{L}y)$, where $y\in \mathbb{R}^{l}$.
\end{definition}

\begin{example} \label{ex:appredsys}
	Consider the system in Equation \eqref{eq:examplepert} and the matrix $L$ given in Example \ref{ex:firstode}.
	We can compute the reduced system as follows
	\begin{equation*}
		\begin{pmatrix}
			\dot{y}_{1} \\
			\dot{y}_{2}
		\end{pmatrix}
		= L f \left(  \bar{L}        \begin{pmatrix}
			y_{1} \\
			y_{2}
		\end{pmatrix}
		\right)
		=
		L f \left(
		\begin{pmatrix}
				y_{1} & 0        \\
				0     & 0.2y_{2} \\
				0     & 0.4y_{2}
			\end{pmatrix} \right)
		=
		\begin{pmatrix}
			\frac{1.004y_{2}^{2}}{y_{1}^{2}+1} \\
			\frac{-2y_{2}}{y_{2}+1}
		\end{pmatrix}.
	\end{equation*}
\end{example}

\begin{definition}\label{def:error}
	Given the system in Equation~\eqref{eq:ode} and an approximate lumping $L$, we define the \emph{error} of the reduction by $ e(t):= y(t) - Lx(t),$ where $y$ is the solution of the reduced system given by Definition \ref{def:redsys}.
	The evolution of the approximation error is given by $\dot{e} = \dot{y} - L \dot{x}$, with $e(0) = 0$.
\end{definition}

\begin{example}
	Following Example \ref{ex:appredsys}, we compute the trajectories of the reduced and original systems.
	We also compute the $L_{2}$-error of the reduction.
	This is summarized in Figure \ref{fig:apperr}.
	Note that, in this example, we obtained a reduction of a system that was not exactly lumpable.
	In Figure \ref{fig:apperr:sim}, we can appreciate that, for the given time horizon, the reduced and original values are close to each other.
\end{example}

\begin{figure}[t]
	\centering
	\begin{subfigure}[t]{0.45\textwidth}
		\centering
		\includegraphics[width=0.8\textwidth]{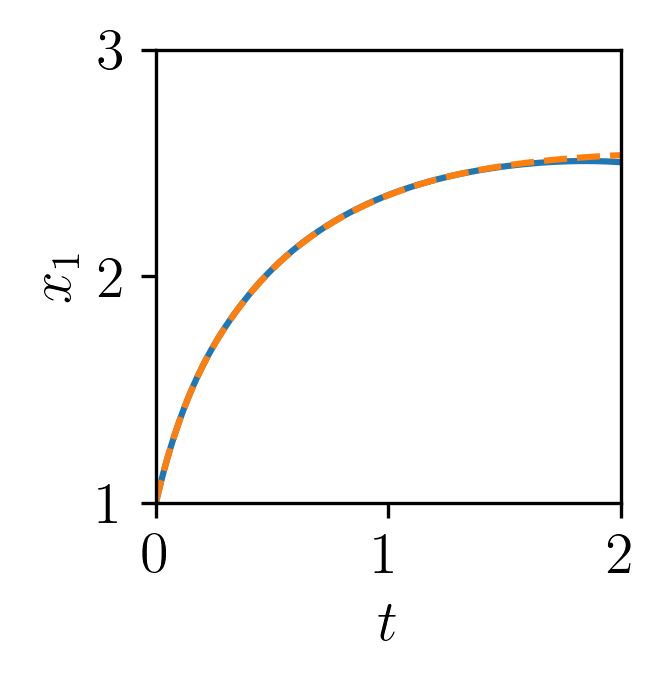}
		\caption{Evolution of the reduced (orange) and original (blue).}
		\label{fig:apperr:sim}
	\end{subfigure}
	\hspace{1cm}
	\begin{subfigure}[t]{0.45\textwidth}
		\centering
		\includegraphics[width=0.9\textwidth]{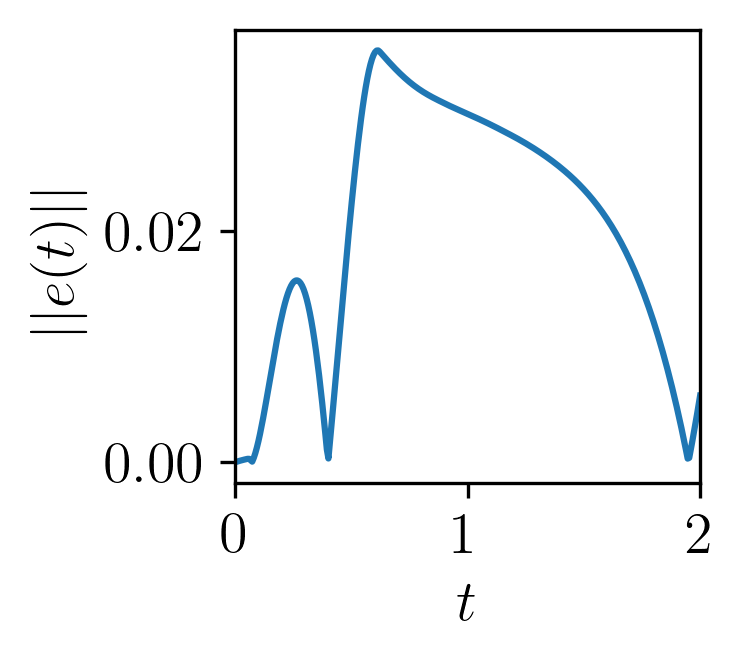}
		\caption{Error evolution $\norm{e(t)}$ }
	\end{subfigure}
	\caption{Reduced system and error computation for Example \ref{ex:appredsys} using the matrix $L$ of Example \ref{ex:firstode} on Equation \eqref{eq:examplepert}.}
	\label{fig:apperr}
\end{figure}

We next show that it is possible to bound the error introduced by approximate constrained lumping.
Even though our worst-case bound is often conservative in practice, the bound confirms the consistency of the approach: the error is of order $\mathcal{O}(\eta)$ --- that is, the actual error $e(t)$ decreases linearly with the deviation $\eta$.
Moreover, as we will see in Section~\ref{sec_eval}, our framework can find approximate lumpings with low errors for published biological models.
\begin{theorem}[Error Bounds]\label{thm:errbound}
	Fix a bounded set of initial conditions $S$, a finite time horizon $T$ and assume that the respective reachable set of the $m$-dimensional ODE system $\dot{x} = f(x)$ is bounded on $[0,T]$.
	If there exists a compact set $\Omega$ such that $x(t), \bar{L}Lx(t) \in \Omega$ for all $t \in [0,T]$ and all $x(0)\in S$ and $f$ is analytic in $\Omega$, then for any $\eta > 0$ for which $L$ is a $(S,T, \eta)-$ lumping, it holds that $\norm{e(t)} \leq \eta \cdot K_{C,L,T}$, where
	\[
		K_{C,L,T} = \frac{1}{C \norm{L} \norm{\bar{L}}} \left( e^{C \norm{L} \norm{\bar{L}}T}-1 \right).
	\]
	Here, $C$ is the Lipschitz constant of $f$ over the set of initial conditions $S$.
\end{theorem}

\subsection{Lumping Algorithm}\label{sec:algacl}

In this section we relax the condition in Line~\ref{alg:findL:check} of Algorithm~\ref{alg:findL}, thus allowing us to find approximate reductions. 
Intuitively, we add a new row $rJ_{i}$ to the matrix $L$ only if it is \textit{far enough} from $\rsp(L)$.
We make this rigorous by fixing a lumping tolerance $\varepsilon$ and adding a row $rJ_{i}$ only if 
$\norm{rJ_{i} - \pi_{i}} > \varepsilon$, where $\pi_{i}:=rJ_{i}P_{L}$ is the orthogonal projection of $rJ_{i}$ onto $\rsp(L)$.
Thus, we propose Algorithm \ref{alg:acl}.

\begin{algorithm}[t]
 \caption{Approximate Constrained Lumping Algorithm}\label{alg:acl}
 \begin{algorithmic}[1]
  \REQUIRE numerical threshold $\varepsilon \geq 0$, and 
  a set of matrices $\left\{ J_1,\dots,J_N \right\}$ spanning $\mathcal{V}_{J}$ (Theorem \ref{thm:repJ}), and 
  a $p \times m$ matrix of observables $M$ with row rank $p$

  \STATE \textbf{compute} orthonormal rows spanning the row space of $M$ and store them as $M$

  \STATE \textbf{set} $L := M$\label{alg:acl:setL0}

  \REPEAT
  \FORALL{$1 \leq i \leq N$ and rows $r$ of $L$} \label{alg:acl:mainloop}
  \STATE {\textbf{compute} $\pi_{i} := r J_i L^T  L$} \label{alg:acl:proj}
  \IF{$\norm{r J_i - \pi_{i}} > \varepsilon$}\label{alg:acl:check}
  \STATE {\textbf{append} row $(r J_i - \pi_{i}) / \norm{r J_i -\pi_{i} }$ to $L$}\label{alg:acl:append}
  \ENDIF
  \ENDFOR
  \UNTIL{no rows have been appended to $L$}

  \RETURN lumping matrix $L$.

 \end{algorithmic}
\end{algorithm}
\begin{figure}
    \centering
 \includegraphics[width=0.4\linewidth]{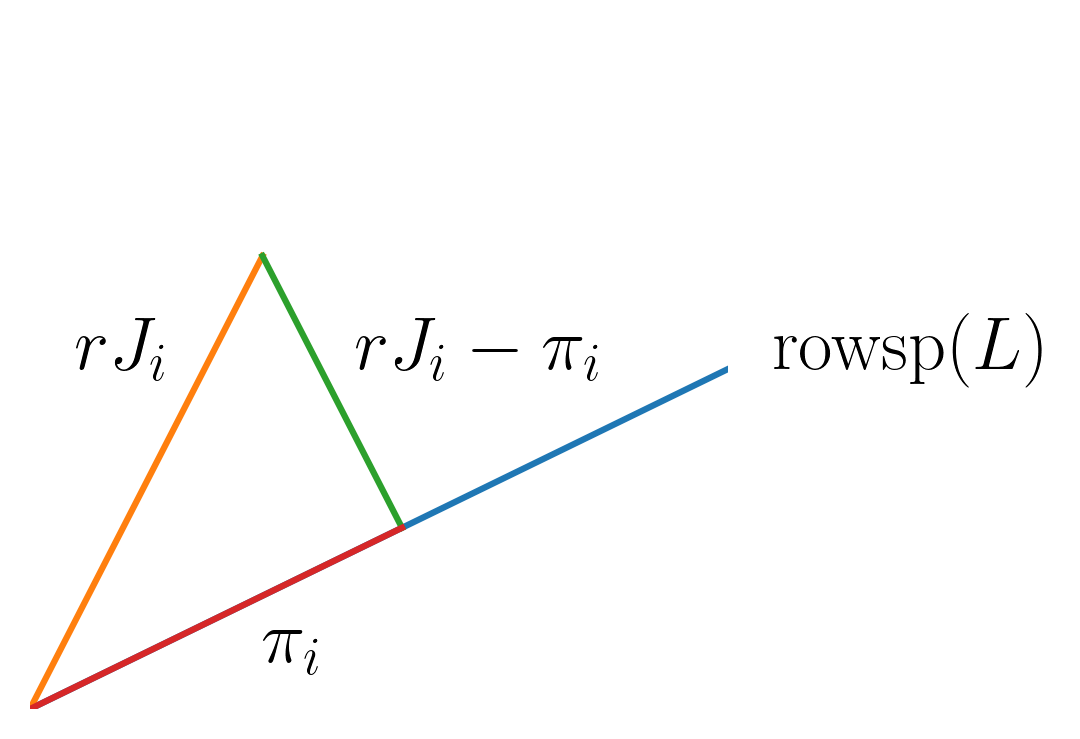}
 \caption{Decomposition of $rJ_i$ into $\rsp(L)$ and $\rsp(L) ^{\perp}$ \cite{leguizamon-robayo_approximate_2023}.}
 \label{fig:newrow}
\end{figure}
%
\begin{remark}
 In Line \ref{alg:acl:append} of Algorithm \ref{alg:acl} we do not append $rJ_{i}$, given by the orange vector in Figure \ref{fig:newrow}.
 Rather, we append the normalized component of $rJ_i$ in the orthogonal direction to $\rsp(L)$, which is $r J_i - \pi_{i}$, the green vector in Fig.~\ref{fig:newrow}.
 This ensures that the matrix $L$ is orthonormal.
 We can use the fact that the pseudoinverse of an orthonormal matrix is its transpose to obtain that $\bar{L}=L^{T}$.
\end{remark}

We now provide a detailed example of Algorithm~\ref{alg:acl}.

\begin{example}
 Consider the system given by Equation~\eqref{eq:examplepert} and let $\varepsilon=0.2$.
 Set $M=(1,0,0)$ i.e., we are observing the component $x_{1}$ from the original system.
 Using Algorithm~\ref{alg:findJ} with the same points as in Example~\ref{ex:findJ}, we find that the dimension of $\mathcal{V}_{J}$ is 6 with basis given by the matrices:

  \begin{equation*}
   \begin{split}
    J_1 &=
    \begin{pmatrix}
     0  & 0  & 0  \\
     2  & -2 & -8 \\
     -1 & 0  & 2
    \end{pmatrix},\\
    J_3 &=
    \begin{pmatrix}
     0      & 4.05   & 8      \\
     0.667  & 0.444  & -0.444 \\
     -0.333 & -0.333 & 0
    \end{pmatrix}, \\
    J_5 &=
    \begin{pmatrix}
     -2.958    & 1.41     & 2.815 \\
     0.25    & 0.031    & -0.438   \\
     -0.125 & -0.031 & 0.188
    \end{pmatrix},
   \end{split}
   \hspace{1cm}
   \begin{split}
   J_2 &=
    \begin{pmatrix}
     0    & 2     & 4.05 \\
     1    & 0     & -2   \\
     -0.5 & -0.25 & 0.5
    \end{pmatrix},\\
    J_4 &=
    \begin{pmatrix}
     -40.75 & 9.05   & 18.125 \\
     0.200  & 0.060  & -0.280 \\
     -0.100 & -0.040 & 0.120
    \end{pmatrix},\\
    J_6 &=
    \begin{pmatrix}
     -0.951    & 0.621     & 1.235 \\
     0.222    & 0.025     & -0.395   \\
     -0.111 & -0.025 & 0.173
    \end{pmatrix}.
   \end{split}
  \end{equation*}

 We begin the computation in Line \ref{alg:acl:setL0} by setting $L=M=(1,0,0)$ and begin the main loop in Line \ref{alg:acl:mainloop} with $r=(1,0,0)$.
 To carry out the computation of Line \ref{alg:acl:proj},  we have that 
 \begin{equation*}
  \begin{split}
   rJ_{1}&=(0,0,0),   \\
   rJ_{3}&=(0,4.05,8),\\
   rJ_{5}&=(0   , 1.41 , 2.815),
  \end{split}
  \hspace{1cm}
  \begin{split}
   rJ_{2}&=(0,2,4.05),   \\
   rJ_{4}&=(0, 9.05, 18.125), \\
   rJ_{6}&= (0   , 0.621, 1.235).
  \end{split}
 \end{equation*}
 Let $d_{i}:=\norm{r J_i - \pi_{i}}$.
 By the check of Line \ref{alg:acl:check}, we do not append any new row to $L$ as $d_{1}=0$ since $rJ_{1}=\pi_{1} = (0,0,0)$.
 Next, we compute $\pi_{2}$ following Line \ref{alg:acl:proj}.
 As $\pi_{2} =0$, it follows that $d_{2} = \norm{rJ_{2}}=4.52>0.2$, which, by Line \ref{alg:acl:check}, means that we need to append a new row to $L$.
 By Line \ref{alg:acl:append}, we add normalized $rJ_{2} - \pi_{2}$ as a row of $L$ (Line \ref{alg:acl:append}), thus obtaining
 \begin{equation}\label{eq:resultingL}
  L = \begin{pmatrix}
   1 & 0     & 0     \\
   0 & 0.443 & 0.897
  \end{pmatrix}.
 \end{equation}
 Going back to Line \ref{alg:acl:proj} with the updated matrix $L$, we compute $\pi_{3} = (0, 3.974, 8.046)$.
 Following Line \ref{alg:acl:check}, we do not add additional rows to $L$ since $d_{3} = 0.089 < 0.2$.
 We compute $\pi_{4}= (-40.75, 8.978, 18.18)$ and $d_{4} = 0.09 < 0.2 $, leaving $L$ unchanged.
 We go on with the computations for the remaining matrices $J_{i}$ to obtain $d_{5}=0.018,~d_{6}=0.01$.
 It follows that we need not add any new rows to $L$.

 So far we have only checked the first row of $L$.
 Following the main loop (Line \ref{alg:acl:mainloop}), we set $r = (0, 0.443, 0.897)$.
 We compute
 \begin{equation*}
  \begin{split}
   rJ_{1}&=(-0.011,-0.869,-1.760),   \\
   rJ_{3}&=(-0.004,-0.098,-0.199),\\
   rJ_{5}&=(-0.001, -0.013, -0.026),
  \end{split}
  \hspace{1cm}
  \begin{split}
   rJ_{2}&=(-0.006,-0.218,-0.441),   \\
   rJ_{4}&=(-0.001, -0.008, -0.017), \\
   rJ_{6}&= (-0.001, -0.010 , -0.021).
  \end{split}
 \end{equation*}
 We continue computing $d_{1} = 0.02$, $d_{2} = 0.007,~d_{3}= 0.004$, while $d_{i}=0.001$ for $i= 4,5,6$. 
    It follows that the check of Line~\ref{alg:acl:check} is false for all $i=1,\dots,6$; meaning that no more rows should be added to $L$, thus terminating the algorithm.
 The output of Algorithm \ref{alg:acl} is the matrix $L$ of Equation \eqref{eq:resultingL}.
\end{example}

 The following result states that approximate constrained lumpings can be efficiently computed.

\begin{theorem}[Time Complexity of Algorithm~\ref{alg:acl}]\label{thm:complexity}
Approximate constrained lumping can be computed in polynomial time. 
Specifically, the complexity of Algorithm~\ref{alg:clue} when using Algorithms~\ref{alg:findJ} and \ref{alg:acl} as subroutines can be bounded by $\mathcal{O}((A+Np(p+2))m^2)$, where $A$ is the cost of computing an entry of $f(x)$,  $m$ is the dimension of $f$, $N$ is the number of matrices spanning $\mathcal{V}_{J}$, and $p$ is the dimension of the reduced system.
\end{theorem}

When the system is not approximately lumpable, it follows that the worst-case time complexity is $O(m^4)$.
The complexity bound of Theorem~\ref{thm:complexity} can be further improved by using more efficient data structures~\cite{ovchinnikov_clue_2021}. 

The next result relies on Theorem~\ref{thm:repJ} and ensures that the approximate reductions found by Algorithm~\ref{alg:acl} admit a deviation tolerance of order $\mathcal{O}(\varepsilon)$. In other words, the deviation tolerance $\eta$ is in the order of the lumping tolerance $\varepsilon$.

\begin{theorem}[Correctness of Algorithm~\ref{alg:acl}]\label{thm:algorithm:correctness}
 Using the setting of Theorem~\ref{thm:errbound}, assume $J(x)$ can be written in the form given by Equation~\eqref{eq:jacrep}.
 Let $L$ be the matrix computed via Algorithm \ref{alg:acl} with lumping tolerance $\varepsilon$ over the sampled matrices  $\{J(x_1),\ldots,J(x_N)\}$ for $x_1,\ldots,x_N \in \Omega$.
 Then, $L$ is a $(T, S, \sqrt{m} C'C K\varepsilon)$-lumping, where $C' =\sup_{i, x\in \Omega} \left\lvert \mu_{i}(x) \right\lvert$, $C$ is a constant depending on the representation of $J(x)$ and $K$ is a constant depending on the trajectories $x(t)$ and $\bar{L}Lx(t)$.
\end{theorem}

\subsection{Heuristic search of lumping tolerance}\label{sec:besteps}

To apply Algorithm \ref{alg:acl}, it is important to choose an appropriate value for $\varepsilon$.
While Theorems~\ref{thm:errbound} and~\ref{thm:algorithm:correctness} allow for an error estimation in terms of $\varepsilon$, in practice, this bound is not tight enough.
For this reason, we introduce a heuristic approach to appropriately choose $\varepsilon$.

Intuitively, by increasing $\varepsilon$, it is possible to lump more variables together at the price of larger approximation errors.
We would like to find the largest admissible $\varepsilon$ such that the approximation error is small enough.
To this aim, we begin by noting that the minimum value that $\varepsilon$ can have is $0$, corresponding to an exact reduction.
A naive idea would be to set $\varepsilon=0$ and add small increments until the reduction is satisfying.
However, this requires an appropriate choice of increment which in turn depends on the model.
Instead, Lemma \ref{lmm:uppereps} gives us an upper bound for $\varepsilon$ which can be computed for each model.
\begin{lemma}[Upper bound on $\varepsilon$]\label{lmm:uppereps}
 Consider a  matrix of observables $L\in \mathbb{R}^{p\times m}$ of rank $p$ with $j-$th row denoted by $r_{j}$ and a decomposition of the Jacobian $J(x) = \sum_{i=1}^{n}J_{i}\mu_{i}(x)$.
 Let $\varepsilon_{\max}$ be given by
 \begin{equation} \label{eq:epsmax}
  \varepsilon_{\max} = \max_{j,i} \norm{r_{j}J_{i} - r_{j}J_{i}P_L},
 \end{equation}
    for $i=1,\dots, n$ and each $j=1,\dots, l$, where $P_L$ is the projection onto $\rsp(L)$.
    Then $\varepsilon_{max}$ is the smallest $\varepsilon$ such for any $\varepsilon \geq \varepsilon_{\max}$ the output of Algorithm \ref{alg:acl} will be a matrix of orthonormal rows spanning the row space of $L$.
\end{lemma}

Lemma \ref{lmm:uppereps} states that any $\varepsilon \geq \varepsilon_{\max}$ will collapse the dynamics of the system onto observable $M$.
In this way, for any given model, we have a range $[0;\varepsilon_{\max}]$ from which to choose the right value for $\varepsilon$.

\begin{remark}\label{rmk:epsmaxrat}
The value of $\varepsilon_{\max}$ depends not only on the chosen model and observable but also on the set of matrices $\mathcal{V}_{J}$ used to represent $J(x)$. 	 Compared to polynomial models, the values for the drift in rational models can be more sensitive to changes in $x$, especially when the denominator of $f(x)$ is close to $0$. This sensitivity can be explained by vanishing denominators in the Jacobian $J(x)$ of a rational drift $f(x)$. We therefore expect rational models to exhibit larger fluctuations of $\varepsilon_{\max}$.
\end{remark}

For example, the evolution of the reductions for different values of $\varepsilon$ is demonstrated for a model from the literature in Figure~\ref{fig:staircase}. The model, named BIOMD103, has been downloaded from the online repository ODEbase~\cite{LuedersSturmRadulescu22} discussed in greater detail in Section~\ref{sec_eval} (it is model 2 from Table \ref{tab:chosenmod}).

\begin{figure}[h]
 \centering
 \includegraphics[width=0.5\linewidth]{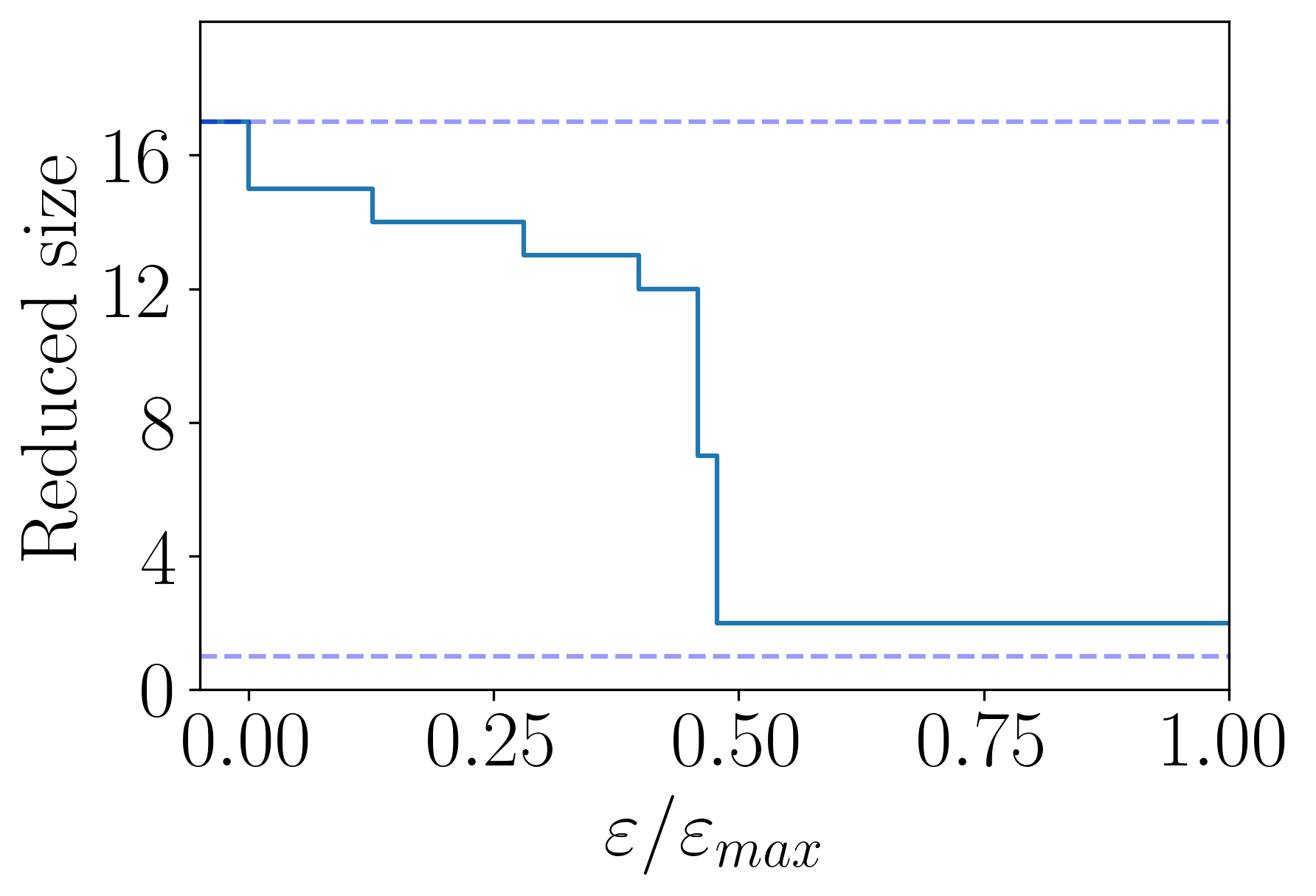}
 \caption{Reduced size vs relative $\varepsilon$ for the model BIOMD103.
 }
 \label{fig:staircase}
\end{figure}

We note that previous work~\cite{leguizamon-robayo_approximate_2023} used the deviation as a proxy to estimate the error without having to simulate the original system. This approach relied on a Monte Carlo approximation of the $L_1$ norm of $\dev_L$ over $[0, \norm{x^{0}}]^m$, where $x^{0}$ is the center of the initial set.
While being effective in finding meaningful reductions for polynomial systems, this deviation-based approach can be problematic for rational systems. 
This is because rational systems are more sensitive to the choice of compactum, since a small change in the denominator can have a large effect on the overall value of $f(x)$, especially for values of $x$ where the denominator approaches $0$.

We propose a novel approach which gives satisfactory results for both polynomial and rational systems  (Section~\ref{sec_eval}).
Using the fact that the size of a reduced model decreases monotonically with $\varepsilon$, we set up a target (a \textit{cutoff} size, $m^{*}$) for the reduced model given as a percentage of its original size.
We then apply a binary search algorithm to find the largest $\varepsilon$  such that the reduced model's size $m_{\varepsilon}$ satisfies $m_{\varepsilon} \leq m^{*}$.
In terms of Figure \ref{fig:staircase}, we want to obtain the smallest $\varepsilon$ whose induced size is below the cutoff size.
This intuition is formalized by Algorithm \ref{alg:acleps}.

\begin{algorithm}[h]
 \caption{Finding the smallest acceptable $\varepsilon$ for Algorithm \ref{alg:acl}
 }
 \begin{algorithmic}[1]\label{alg:acleps}

  \REQUIRE set of matrices $J_{i}, ~i=1,\dots,N$ spanning $\mathcal{V}_{J}$, and 
   cutoff reduced size $m^{*}$, and
  minimal difference $d_{\min}$ 
  \STATE {\textbf{set} $\varepsilon_{\min}=0$}
  \STATE {\textbf{compute} $m_{\min} = nrows(L_{\varepsilon_{\min}})$, \\ \hspace{0.2cm}
  $\varepsilon_{\max} = \max_{j,i} \norm{r_{j}J_{i}}$,\\ \hspace{0.2cm}
  $m_{\max} = nrows(L_{\varepsilon_{\max}})$}

  \IF{$m^{*}<m_{\min}$} \label{alg:acleps:check}
  \RETURN $\varepsilon_{\min}$
  \ELSIF{$m^{*}>m_{\max}$}
  \RETURN $\varepsilon_{\max}$
  \ENDIF

  \REPEAT \label{alg:acleps:mainloop}
  \STATE {\textbf{compute} $\varepsilon = (\varepsilon_{\max}+\varepsilon_{\min})/{2}$}
  \STATE {\textbf{compute} $L$ using Algorithm~\ref{alg:acl} with $\varepsilon$ }
  \IF{ $nrows(L) < m^{*} $ }
  \STATE {\textbf{set} $\varepsilon_{\max} = \varepsilon$ }
  \ELSE 
  \STATE {\textbf{set} $\varepsilon_{\min} = \varepsilon$ }
  \ENDIF
  \STATE {\textbf{compute} $d= \varepsilon_{\max}-\varepsilon_{\min}$}
  \UNTIL{$d < d_{\min}$}

  \RETURN $\varepsilon_{\max}$.
 \end{algorithmic}
\end{algorithm}

\begin{remark}\label{rmk:dtradeoff}
It is important to notice that the choice of minimal different $d_{min}$ in Algorithm~\ref{alg:acleps} 
provides a trade-off between detecting reductions and the speed of the approach. 
A lower value for $d_{min}$ can lead to finding more reductions at the cost of more iterations. 
\end{remark}
%
\begin{example}\label{ex:findeps}
 Let us use again model BIOMD103 from Figure~\ref{fig:staircase}.  It has an observable of interest, $C3$~\cite{legewie_mathematical_2006}, consisting of one variable representing the concentration of\emph{Activated capsase3}. 
 We would like to find a reduction such that the reduced model size is at most $62.5\%$ of the original size ($17$), i.e., $11$ species.
 To do so, we use Algorithm \ref{alg:acleps}.
 We first verify (Line~\ref{alg:acleps:check}) that the cutoff size is larger than the minimum size and smaller than the size of the exact reduction.
 As the cutoff size $11$ is between $1$ and $17$ we continue to Line~\ref{alg:acleps:mainloop}.
 On the first iteration (Figure \ref{fig:findeps:1}) 
 the middle point (green) gives a reduction of size $2$.
 As this size is smaller than the goal size $11$, we set $\varepsilon_{\max} = \varepsilon$.
 This gives a new interval to search for a reduction (black line on Figure \ref{fig:findeps:2}).
 We continue following the algorithm until it converges to a reduction of size $7$ (green line Figure \ref{fig:findeps:3}).
 This is a correct result as we can see that the closest step below the cutoff size is the one of size $7$ in Figure \ref{fig:findeps:1}.
\end{example}

\begin{figure}[t]
 \centering
 \begin{subfigure}[b]{0.32\textwidth}
  \includegraphics[width=\textwidth]{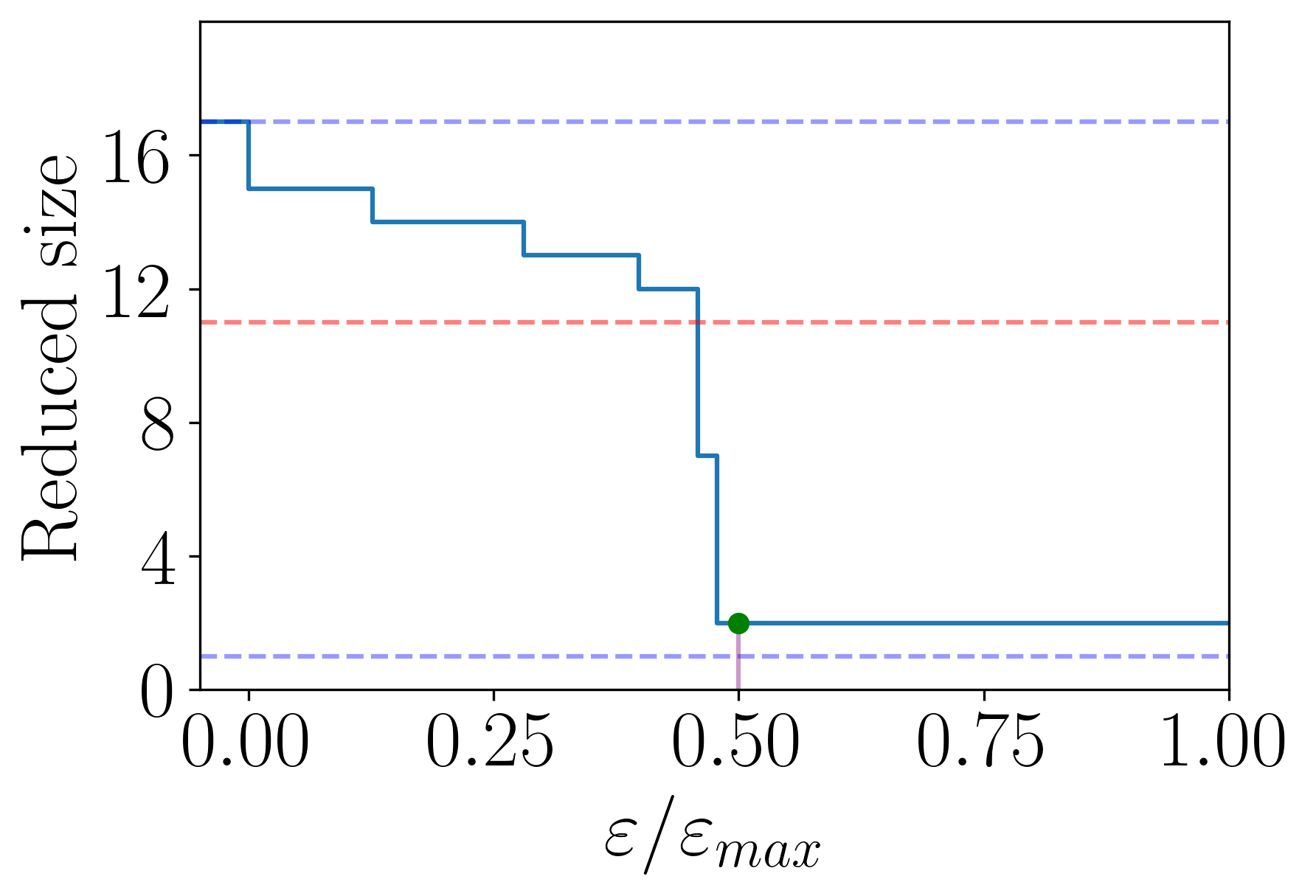}
  \vspace{-0.5cm}
  \caption{First iteration} \label{fig:findeps:1}
 \end{subfigure}
 \begin{subfigure}[b]{0.32\textwidth}
  \includegraphics[width=\textwidth]{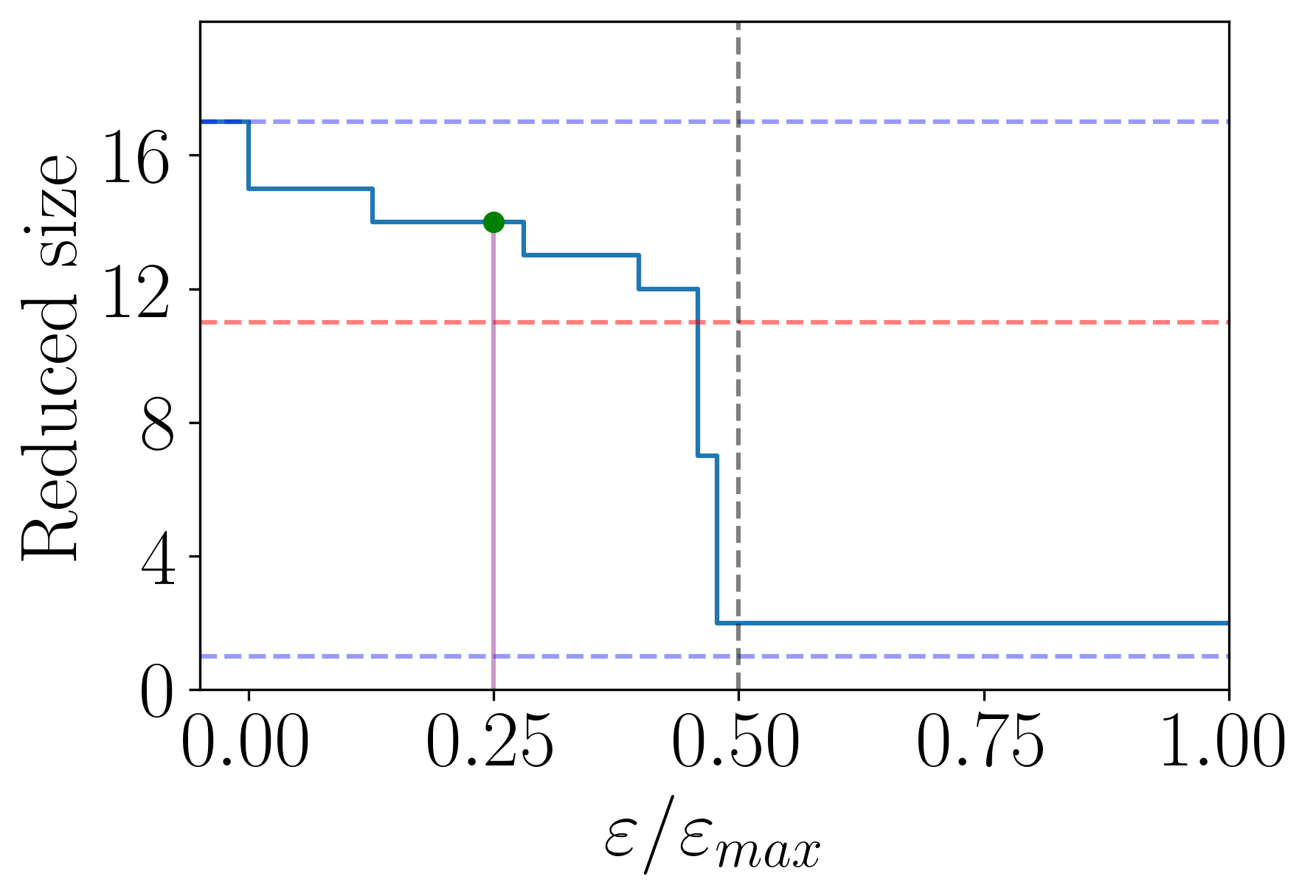}
  \vspace{-0.5cm}		
  \caption{Second iteration} \label{fig:findeps:2}
 \end{subfigure}
 \begin{subfigure}[b]{0.32\textwidth}
  \includegraphics[width=\textwidth]{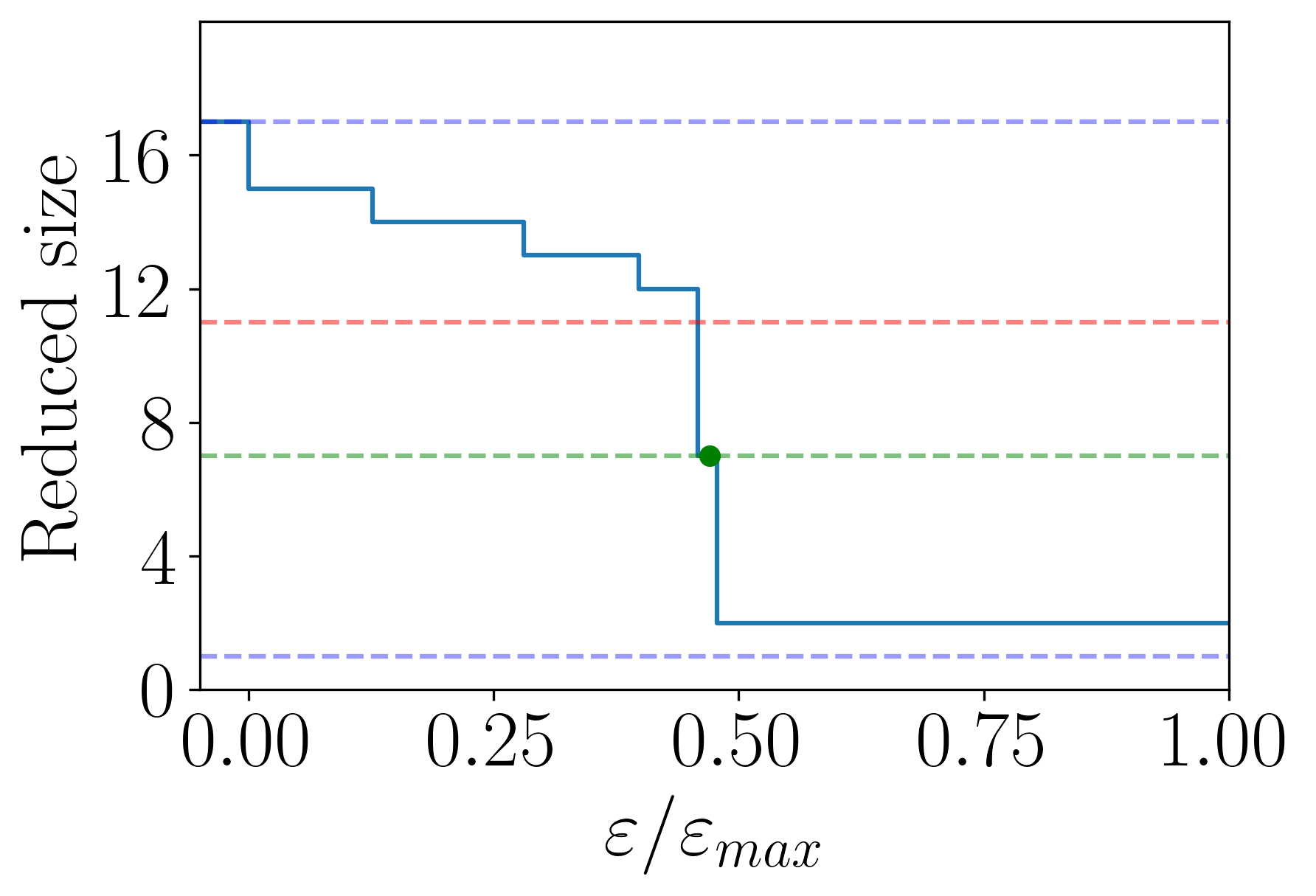}
  \vspace{-0.5cm}		
  \caption{Last iteration} \label{fig:findeps:3}
 \end{subfigure}
  \vspace{-0.3cm}	
 \caption{Algorithm \ref{alg:acleps} run
  on model $2$ to reduce it to size at most $11$ (red line).}
\end{figure}

\section{Evaluation}\label{sec_eval}

In this section, we evaluate our method.
All experiments were performed using an extension of the tool for CLUE~\cite{ovchinnikov_clue_2021,jimenez_clue_2022, cmsb2024_clue}, an open-source Python library.
The experiments are available in the following repository.
 \begin{center}
         \texttt{\url{https://github.com/clue-developers/CLUE/releases/tag/v1.7.3}}
    \end{center}
    In Section~\ref{sec_eval:red_pow}, we use a selection of models from the literature to show how our approach can lead to small reductions while not incurring a large error.
    In Section~\ref{sec_eval:scal}, we use a perturbed multisite protein phosphorylation model to show the scalability of our approach.

\subsection{Reduction power}\label{sec_eval:red_pow}
 To evaluate the reduction power of our approach, we selected several models from the literature.
 These models are chemical reaction networks with $5$ or more species.
 We used two sources to select models.	
 The first is a collection of polynomial models from the literature written in the input format of the tool BioNetGen~\cite{Blinov22112004}. 
    We selected two such models.
 The second is the so-called ODEbase project~\cite{LuedersSturmRadulescu22} which offers importing/exporting capabilities for polynomial and rational models from the BiomodelsDB repository~\cite{BioModels2010}. 
    We selected six such models. 
 Models with prefix \texttt{BIOMD} were taken from ODEbase, while the remaining models were taken directly from the supplementary material of their original paper.
 For each model, the observable was chosen according to the observables studied in the original paper describing the model.
 These typically consist of one variable or linear combinations of them.
 For models where more than one observable was studied in the original paper, we arbitrarily selected one of them.
 Models written in the BioNetGen format were first imported by ERODE~\cite{cardelli_erode_2017}, our tool that collects many techniques for the reduction of biological systems, and then exported in ERODE's \texttt{.ode}  format supported by CLUE. 

To import models from the ODEbase project~\cite{LuedersSturmRadulescu22}, instead, we directly implemented in CLUE an importer to convert them in \texttt{.ode} format. 

 When available, the set of initial conditions necessary for simulations was taken from the original paper or the model files in BiomodelsDB~\cite{BioModels2010}.
 The time horizon was either taken from the original reference or chosen experimentally by noticing when steady state was reached.
 All this additional information was manually added to the \texttt{.ode} files to the corresponding model.

 When choosing the models, we aimed for a selection that could equally display our approach for both rational and polynomial models.
 The final selection of models is displayed in Table~\ref{tab:chosenmod}.
 The column \emph{Type} reports whether the model is polynomial (P) or rational (R).

 The source paper presenting the model is displayed in the \emph{Ref.} column.
 The column \emph{Description} contains a summary of the biological system studied in the model.
 In column \emph{Observable}, we can see the species (or linear combinations of them) chosen as a constraint for the lumping.
 Finally, column \emph{Size} displays the number of model variables added in the observable over the total number of variables in the system (e.g., in Model 5, S2P sums 10 out of the 24 variables of the model).

\begin{table}[t]
	\centering
	\resizebox{0.9\textwidth}{!}{
		\begin{tabular}{ccccp{3cm}p{2cm}r}
			\hline
			Nr.      & Name           & Type                             & Ref.                                                       & Description                                                            & Observable                               & Size  \\
			\hline
			1        & BIOMD102       & P                                & \cite{legewie_mathematical_2006}                           & Signaling pathways involved in the initiation of apoptosis (wild-type model)            & Activated capsase3  (C3)                 & 1/12  \\
			2        &
			BIOMD103 & P              & \cite{legewie_mathematical_2006} & Signaling pathways involved in the initiation of apoptosis (competitive model) & Activated capsase3  (C3)                                               & 1/17                                             \\
			3        &
			BIOMD447 & P              & \cite{venkatraman_plasmin_2012}  & Thrombospondin-Dependent Activation of TGF-$\beta$1        & Transforming growth factor (TGF)-$\beta$1  (TGFb1)                     & 1/13                                             \\
			4        & BioNetGen\_CPP & P                                & \cite{mu_carbon-fate_2007}                                 & Central carbon pathway of E. coli                                      & 1,3-diphosphoglycerate (D13PG)           & 2/87  \\
			5        & NIHMS80246\_S6 & P                                & \cite{borisov_domain-oriented_2008}                        & FceRI-like network of a cell-surface receptor                          & Phosphorylation at Y2 (S2P)              & 10/24 \\
			6        & BIOMD437       & R                                & \cite{tseng_comprehensive_2012}                            & Circadian clock of the fungus Neurospora crassa                        & Frequency of $mRNA$   (frq\_RNA)         & 1/39  \\
			7        & BIOMD448       & R                                & \cite{brannmark_insulin_2013}                              & Insulin Signaling in Type 2 Diabetes                                   & Phosphorylation of site S6    (S6P)      & 1/23  \\
			8        & BIOMD488       & R                                & \cite{proctor_investigating_2013}                          & Effect of A$\beta$-immunization in Alzheimer's disease & Microtubular binding protein tau   (Tau) & 1/68  \\
			\hline
		\end{tabular}
	}
	\caption{
			List of selected models used for evaluation.
	}
	\label{tab:chosenmod}
\end{table}

For each model, we used Algorithm~\ref{alg:acleps} with $d_{min} = 10^{-6}$ to find the largest lumping tolerance $\varepsilon$ such that the size of the reduced model is larger than a given bound. 
The bound was given as a ratio of the size of the original system. We used ratios from  $12.5\%$ to $87.5\%$ using a step size of $12.5\%$. 
This means that, in the case of Model 2, the initial size was 17 and so the cutoff sizes were $[2.125, 4.25, 6.375, 8.5, 10.625, 12.75, 14.875]$. 
Notably, different ratios can lead to the same reduction. Consider Model~5. 
Here, the exact reduction is $34.5\%$ of the original size, so the output of Algorithm~\ref{alg:acleps} is $0$ for all percentages above $34.5\%$ as they recover the exact reduction.
In such cases, only one result is shown.

Table~\ref{tab:greaterredpower} presents the results for all models from Table~\ref{tab:chosenmod}.
There is a row for each model with information constant across experiments on it, including the value for $\varepsilon_{\max}$.
The first column $|$\textit{Red.}$|$ shows the size of the reduced model obtained by Algorithm~\ref{alg:acl}, while column \textit{Red. ratio} shows the ratio between the size of the reduced model and that of the original model as a percentage. 
We display the absolute error and the relative error at the time horizon, respectively, in the columns $e(T)$ and $e(T)_{Rel}$.
Here, the relative error is given as the absolute error divided by the value of the observable of the exact reduction.
Column $e_{max}$ shows the maximum value for $e(t)$ within the time horizon.
The obtained value of $\varepsilon$ and its ratio w.r.t. $\varepsilon_{\max}$  are shown in the columns $\varepsilon$ and $\varepsilon/\varepsilon_{\max}$, respectively.
The column $Iter.$  shows the number of iterations and the average computation time of Algorithm \ref{alg:acleps} over 5 runs. 
We used a 4.7 GHz Intel Core i7 computer with 32 GB of RAM to carry out all computations. 
For each model, we plot in Figure~\ref{fig:greaterredpower} the corresponding simulations of the observables for the considered values of $\varepsilon$. 
We used \texttt{Scipy} with \texttt{RK45} and \texttt{LSODA} solvers to simulate polynomial and rational models respectively.
Table~\ref{tab:greaterredpower} and Figure~\ref{fig:greaterredpower} show the successful use of Algorithm~\ref{alg:acleps} on both polynomial and rational models.

\begin{table}[ht!]
    \centering
    \resizebox{\textwidth}{!}{
        \begin{tabular}{HrrrrrrHr rr}
            \hline
            modelName & \multicolumn{1}{c}{\emph{$|$Red.$|$}}  
            & \multicolumn{1}{c}{\emph{Red. ratio(\%)}} & \multicolumn{1}{c}{\emph{$e_{Rel}(T)$ }} & \multicolumn{1}{c}{\emph{$e(T)$}} & \multicolumn{1}{c}{\emph{$e_{max}$}} & \multicolumn{1}{c}{\emph{$\varepsilon$}} & $\varepsilon_{\max}$ & \multicolumn{1}{c}{\emph{$\varepsilon/\varepsilon_{\max}$ (\%)}} & \multicolumn{1}{c}{\emph{Iter.}} & \multicolumn{1}{c}{\emph{Time (s)}} \\
            \hline
			\multicolumn{11}{c}{1) Model: BIOMD102, size: 13, type: P, exact lumping: 13, $\varepsilon_{\max}:$ 3.52E-03}
			\\
			\hline
            BIOMD102 & 11 & 84.6 & 1.16E+01 & 1.44E+02 & 1.44E+02 & 4.24E-04 & 3.52E-03 & 1.21E+01 & 13 & 0.284 \\
            BIOMD102 & 8 & 61.5 & 4.86E+00 & 6.04E+01 & 6.78E+01 & 9.95E-04 & 3.52E-03 & 2.83E+01 & 13 & 0.099 \\
            BIOMD102 & 6 & 46.2 & 1.26E+00 & 1.56E+01 & 7.39E+01 & 1.73E-03 & 3.52E-03 & 4.91E+01 & 13 & 0.138 \\
            BIOMD102 & 2 & 15.4 & 1.50E+01 & 1.86E+02 & 1.86E+02 & 1.97E-03 & 3.52E-03 & 5.60E+01 & 13 & 0.046 \\
            BIOMD102 & 1 & 7.7 & 1.50E+01 & 1.86E+02 & 1.86E+02 & 3.52E-03 & 3.52E-03 & 1.00E+02 & 13 & 0.052 \\
			\hline
			\multicolumn{11}{c}{2) Model: BIOMD103, size: 17, type: P, exact lumping: 17, $\varepsilon_{\max}:$ 3.52E-03}
			\\
			\hline
            BIOMD103 & 14 & 82.4 & 1.02E-02 & 1.72E+00 & 1.72E+00 & 9.89E-04 & 3.52E-03 & 2.81E+01 & 13 & 0.598 \\
            BIOMD103 & 12 & 70.6 & 9.25E-01 & 1.55E+02 & 1.57E+02 & 1.60E-03 & 3.52E-03 & 4.56E+01 & 13 & 0.288 \\
            BIOMD103 & 7 & 41.2 & 8.52E-01 & 1.43E+02 & 1.44E+02 & 1.64E-03 & 3.52E-03 & 4.67E+01 & 13 & 0.218 \\
            BIOMD103 & 2 & 11.8 & 1.82E-01 & 3.06E+01 & 3.06E+01 & 2.00E-03 & 3.52E-03 & 5.67E+01 & 13 & 0.070 \\
			\hline
			\multicolumn{11}{c}{3) Model: BIOMD447, size: 13, type: P, exact lumping: 13, $\varepsilon_{\max}:$ 2.45E+01}
			\\
			\hline
            BIOMD447 & 11 & 84.6 & 1.33E-03 & 2.02E-05 & 2.02E-05 & 4.95E-03 & 2.45E+01 & 2.02E-02 & 26 & 0.173 \\
            BIOMD447 & 9 & 69.2 & 1.35E-02 & 2.04E-04 & 2.04E-04 & 3.41E-01 & 2.45E+01 & 1.39E+00 & 26 & 0.122 \\
            BIOMD447 & 7 & 53.8 & 8.46E-03 & 1.28E-04 & 1.29E-04 & 3.75E-01 & 2.45E+01 & 1.53E+00 & 26 & 0.123 \\
            BIOMD447 & 6 & 46.2 & 6.99E-03 & 1.06E-04 & 7.02E-04 & 1.40E+00 & 2.45E+01 & 5.71E+00 & 26 & 0.105 \\
            BIOMD447 & 3 & 23.1 & 1.22E-01 & 1.86E-03 & 7.31E-03 & 1.98E+01 & 2.45E+01 & 8.10E+01 & 26 & 0.056 \\
			\hline
			\multicolumn{11}{c}{4) Model: BioNetGen\_CCP, size: 87, type: P, exact lumping: 30, $\varepsilon_{\max}:$ 2.83E+00}
			\\
			\hline
            BioNetGen & 30 & 34.5 & 5.86E-06 & 6.11E-05 & 5.22E-04 & 6.74E-07 & 2.83E+00 & 2.38E-05 & 23 & 1.807 \\
            BioNetGen & 10 & 11.5 & 3.14E-01 & 3.27E+00 & 3.27E+00 & 2.67E-01 & 2.83E+00 & 9.44E+00 & 23 & 2.959 \\
			\hline
			\multicolumn{11}{c}{5) Model: NIHMS80246\_S6, size: 24, type: P, exact lumping: 19, $\varepsilon_{\max}:$ 3.13E-01}
			\\
			\hline
            NIHMS80246 & 19 & 79.2 & 3.16E-11 & 3.16E-09 & 2.14E-04 & 5.97E-07 & 3.13E-01 & 1.91E-04 & 20 & 1.370 \\
            NIHMS80246 & 17 & 70.8 & 8.09E-05 & 8.09E-03 & 4.86E+01 & 2.23E-02 & 3.13E-01 & 7.13E+00 & 20 & 1.948 \\
            NIHMS80246 & 14 & 58.3 & 8.09E-05 & 8.09E-03 & 4.89E+01 & 2.34E-02 & 3.13E-01 & 7.49E+00 & 20 & 1.478 \\
            NIHMS80246 & 9 & 37.5 & 1.00E+00 & 1.00E+02 & 1.00E+02 & 3.00E-02 & 3.13E-01 & 9.58E+00 & 20 & 0.809 \\
            NIHMS80246 & 8 & 33.3 & 1.00E+00 & 1.00E+02 & 1.00E+02 & 1.29E-01 & 3.13E-01 & 4.10E+01 & 20 & 0.586 \\
            NIHMS80246 & 1 & 4.2 & 1.00E+00 & 1.00E+02 & 1.00E+02 & 3.13E-01 & 3.13E-01 & 1.00E+02 & 20 & 0.454 \\
			\hline
			\multicolumn{11}{c}{6) Model: BIOMD437, size: 39, type: R, exact lumping: 21, $\varepsilon_{\max}:$ 7.11E+05}
			\\
			\hline
            BIOMD437 & 21 & 53.8 & 0.00E+00 & 0.00E+00 & 0.00E+00 & 0.00E+00 & 7.11E+05 & 0.00E+00 & 1 & 0.044 \\
            BIOMD437 & 12 & 30.8 & 1.58E-04 & 2.00E-05 & 6.66E-05 & 6.47E-07 & 7.11E+05 & 9.09E-11 & 41 & 3.021 \\
            BIOMD437 & 5 & 12.8 & 1.58E-04 & 2.00E-05 & 6.66E-05 & 1.07E+02 & 7.11E+05 & 1.50E-02 & 41 & 3.645 \\
            BIOMD437 & 4 & 10.3 & 3.31E-01 & 4.19E-02 & 4.29E-02 & 3.33E+02 & 7.11E+05 & 4.69E-02 & 41 & 3.237 \\
			\hline
			\multicolumn{11}{c}{7) Model: BIOMD448, size: 27, type: R, exact lumping: 23, $\varepsilon_{\max}:$ 1.00E+00}
			\\
			\hline
            BIOMD448 & 23 & 85.2 & 0.00E+00 & 0.00E+00 & 0.00E+00 & 0.00E+00 & 1.00E+00 & 0.00E+00 & 1 & 0.002 \\
            BIOMD448 & 19 & 70.4 & 7.45E-08 & 6.79E-06 & 2.90E-02 & 1.60E-02 & 1.00E+00 & 1.60E+00 & 21 & 0.912 \\
            BIOMD448 & 1 & 3.7 & 9.21E-01 & 8.39E+01 & 8.39E+01 & 1.00E+00 & 1.00E+00 & 1.00E+02 & 21 & 0.291 \\
			\hline
			\multicolumn{11}{c}{8) Model: BIOMD488, size: 68, type: R, exact lumping: 64, $\varepsilon_{\max}:$ 7.50E+10}
			\\
			\hline
            BIOMD488 & 59 & 86.8 & 2.88E-02 & 1.05E-02 & 2.85E-02 & 1.50E+02 & 7.50E+10 & 2.00E-07 & 58 & 47.800 \\
            BIOMD488 & 50 & 73.5 & 1.97E-01 & 7.20E-02 & 7.20E-02 & 7.50E+04 & 7.50E+10 & 1.00E-04 & 58 & 46.542 \\
            BIOMD488 & 9 & 13.2 & 9.06E+00 & 3.30E+00 & 3.61E+00 & 1.50E+06 & 7.50E+10 & 2.00E-03 & 58 & 7.276 \\
            BIOMD488 & 8 & 11.8 & 9.06E+00 & 3.30E+00 & 3.61E+00 & 6.00E+07 & 7.50E+10 & 8.00E-02 & 58 & 9.505 \\
            \hline
        \end{tabular}
    }
    \caption{Approximate constrained lumping results for the models from Table~\ref{tab:chosenmod}.}
    \label{tab:greaterredpower}
\end{table}

As expected, Algorithm~\ref{alg:acleps} can recover exact reductions, as seen in Models 6 and 7. 
Similarly, having a small cutoff size ($12.5\%$ of the original model size or lower) results in aggressive reductions that collapse the dynamics as seen in the reduction to 8 species of Model 5 (Figure~\ref{fig:greaterredpower:5}). Taken together, these experiments display the following behaviours:

\begin{enumerate}
 \item \emph{Models that do not admit exact reduction may admit approximate reduction.}
       Models 1, 2, and 3 were not exactly lumpable but were approximately lumped.
       For models 1 and 2, the reductions that incurred acceptable errors were modest, at around $80\%$ of the original model size.
          In the case of Model 3, all reductions incurred in low error, with the smallest having only $3$ species.
 \item \emph{Models that admit limited exact reduction, may admit smaller approximate reductions.}
       Models 5, 7, and 8 had modest exact reductions, while they could be further reduced using approximate lumping.
       In this case, most of the approximate reductions have a small error. 
          The improvement compared to the exact case goes from modest, at around $90\%$ of the size of the exact reduction, to as low as $47\%$ of the exactly reduced size (Model 5).

 \item \emph{Models that admit notable exact reduction, may admit even smaller approximate reductions.}
       Models 4 and 6 had already significant exact reductions; nevertheless, they could be further reduced via approximate lumping.
       In this case, Model 4 shows a significant steady-state error for the approximate reduction with $10$ species. 
          For Model 6, it was possible to find a reduction with less than $15\%$ of the original model size while still obtaining simulation results close to the original simulation.
\end{enumerate}

 As expected, for all models the error starts at $0$ and then increases with $t$.
 This increment is barely noticeable when $t$ is close to $0$ for most models.
 While the error increases monotonically in most cases, this does not hold for all reductions.
 e.g., in Figures~\ref{fig:greaterredpower:1} and~\ref{fig:greaterredpower:5}, the reductions with 7 and 14 species get closer to the original simulations as time increases.

 Furthermore, we see that the steady-state error increases with $\varepsilon$ for most models.
 Nonetheless, Figures~\ref{fig:greaterredpower:2} and~\ref{fig:greaterredpower:3} show that it is possible for aggressive reductions to incur lower errors.
    This can be explained by the influence of sampling on the representation of $\mathcal{V}_{J}$, as the check in Line~\ref{alg:acl:check} depends on the norms of each $rJ_i$.
    Overall, for given models, significant reductions without large errors were obtained by Algorithm~\ref{alg:acleps} with a cutoff size of $80\%$ of the original model size. 
    In principle, smaller cutoff sizes may be used to find meaningful reductions; however, this requires a case-by-case analysis. 

    We see that rational models tend to have larger values of $\varepsilon_{\max}$ compared to polynomial ones (see Remark~\ref{rmk:epsmaxrat}). This leads to more iterations of Algorithm~\ref{alg:acleps}.
    The values of $\varepsilon/\varepsilon_{\max}$ that lead to low errors 
    vary widely depending on the model.
    For polynomial models, most reductions happen under the $10\%$ threshold.
    Instead, for rationals the values change significantly, having reductions for $\varepsilon$ lower than $0.01\%$ of $\varepsilon_{\max}$.
    Following Remark~\ref{rmk:dtradeoff}, choosing a larger value for $d$ in Model 8 would have led to fewer iterations and faster times at the cost of potentially missing some low-error reductions as these happened for low values of $\varepsilon$.
    Finally, we note that assumptions of Theorem~\ref{thm:errbound} were satisfied by the experiments since the obtained approximate lumpings $L$ led to positive $L^TLx$  for all positive vectors $x$, while the denominators of rational $f(x)$ vanished only for negative values of $x$.

\begin{figure}[hp!]
    \centering
    \begin{subfigure}[b]{0.45\textwidth}
        \includegraphics[width=\textwidth]{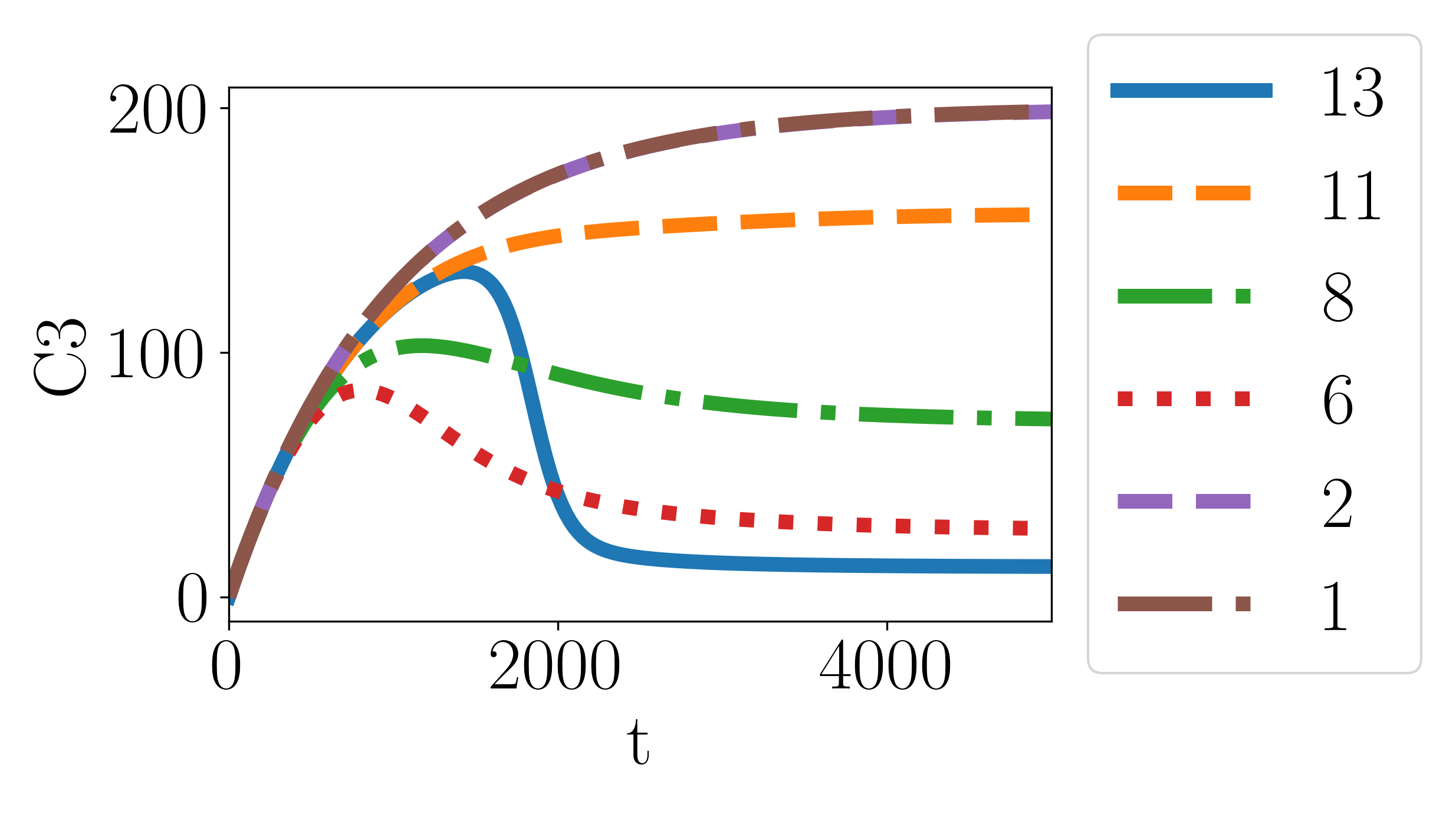}
        \vspace{-0.75cm}
        \caption{Experiment on Model 1} \label{fig:greaterredpower:1}
    \end{subfigure}
    \qquad\quad
    \begin{subfigure}[b]{0.45\textwidth}
        \includegraphics[width=\textwidth]{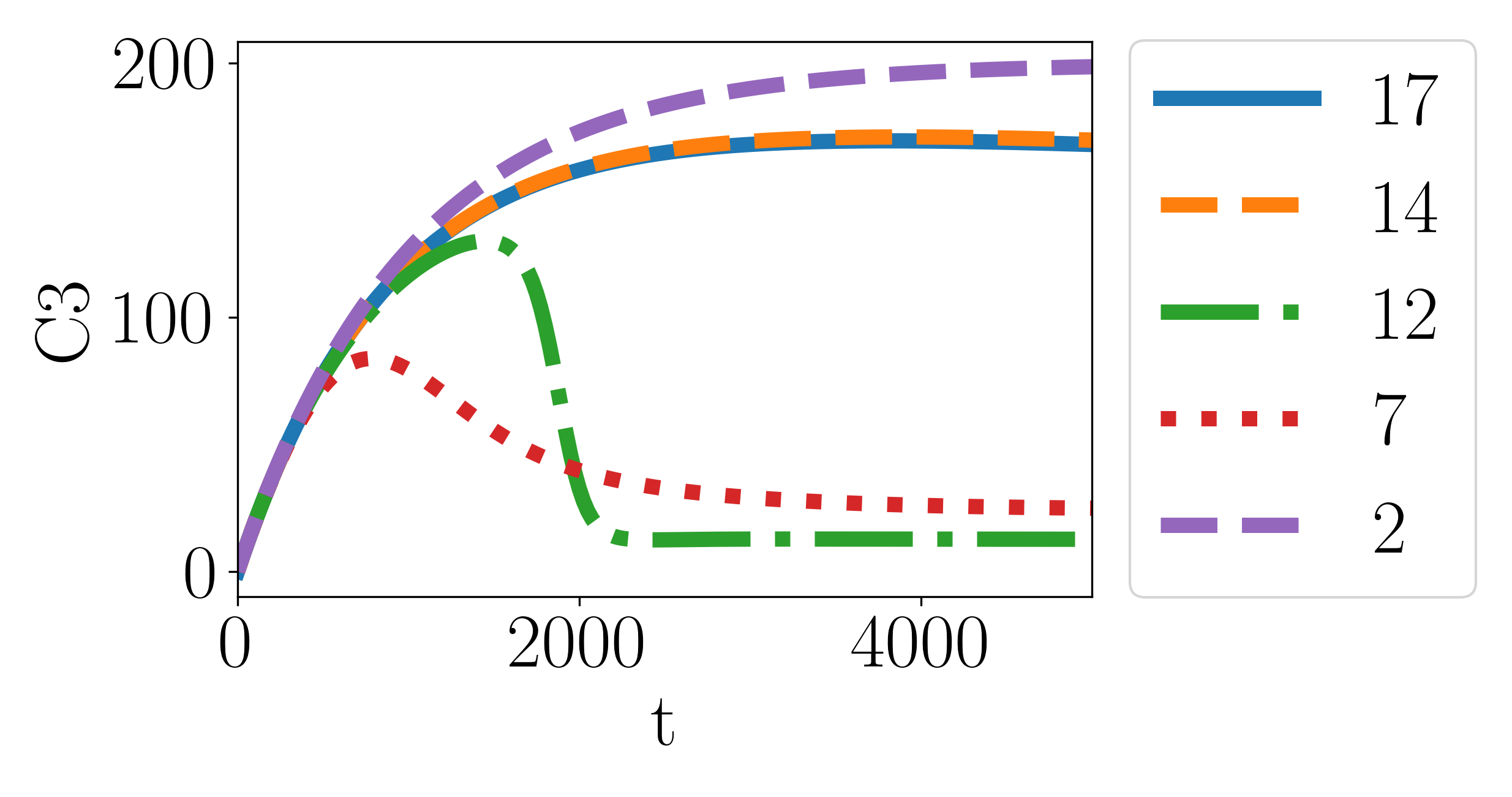}
                \vspace{-0.75cm}
        \caption{Experiment on Model 2}\label{fig:greaterredpower:2}
    \end{subfigure}
    \qquad\quad
    \begin{subfigure}[b]{0.45\textwidth}
        \includegraphics[width=\textwidth]{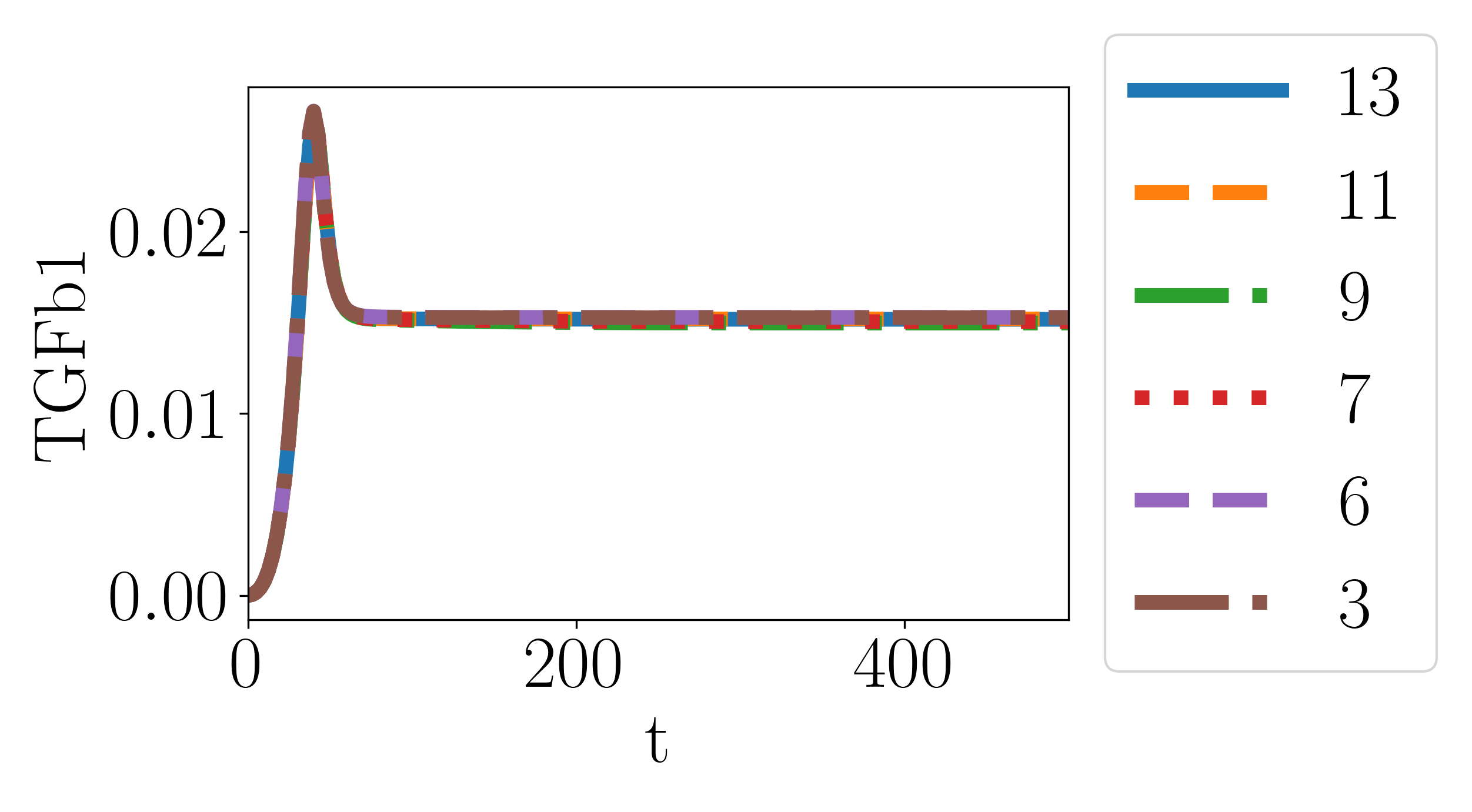}
                \vspace{-0.75cm}
        \caption{Experiment on Model 3}\label{fig:greaterredpower:3}
    \end{subfigure}
    \qquad\quad
    \begin{subfigure}[b]{0.45\textwidth}
        \includegraphics[width=\textwidth]{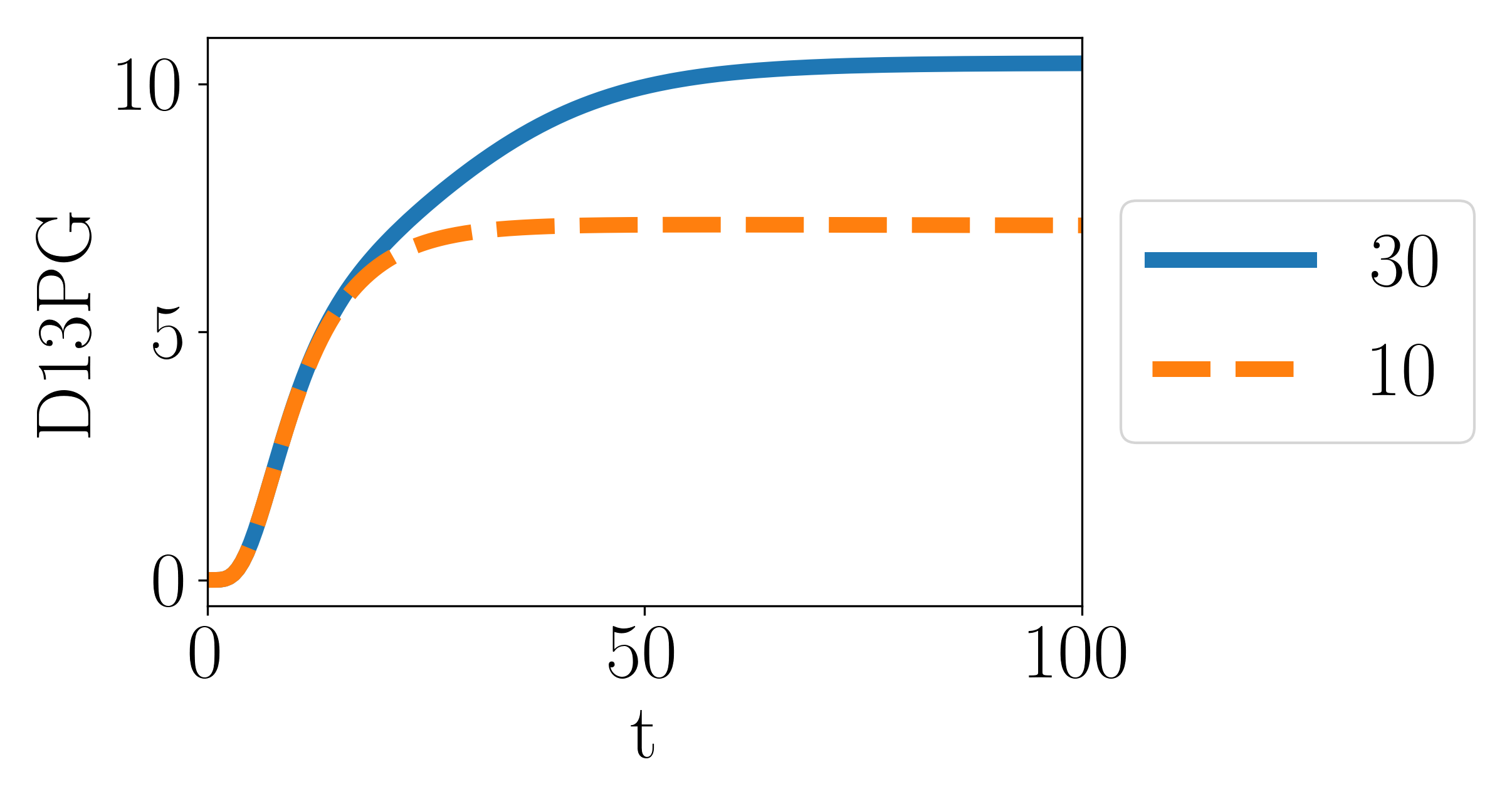}
                \vspace{-0.75cm}
        \caption{Experiment on Model 4}\label{fig:greaterredpower:4}
    \end{subfigure}
    \qquad\quad
    \begin{subfigure}[b]{0.45\textwidth}
        \includegraphics[width=\textwidth]{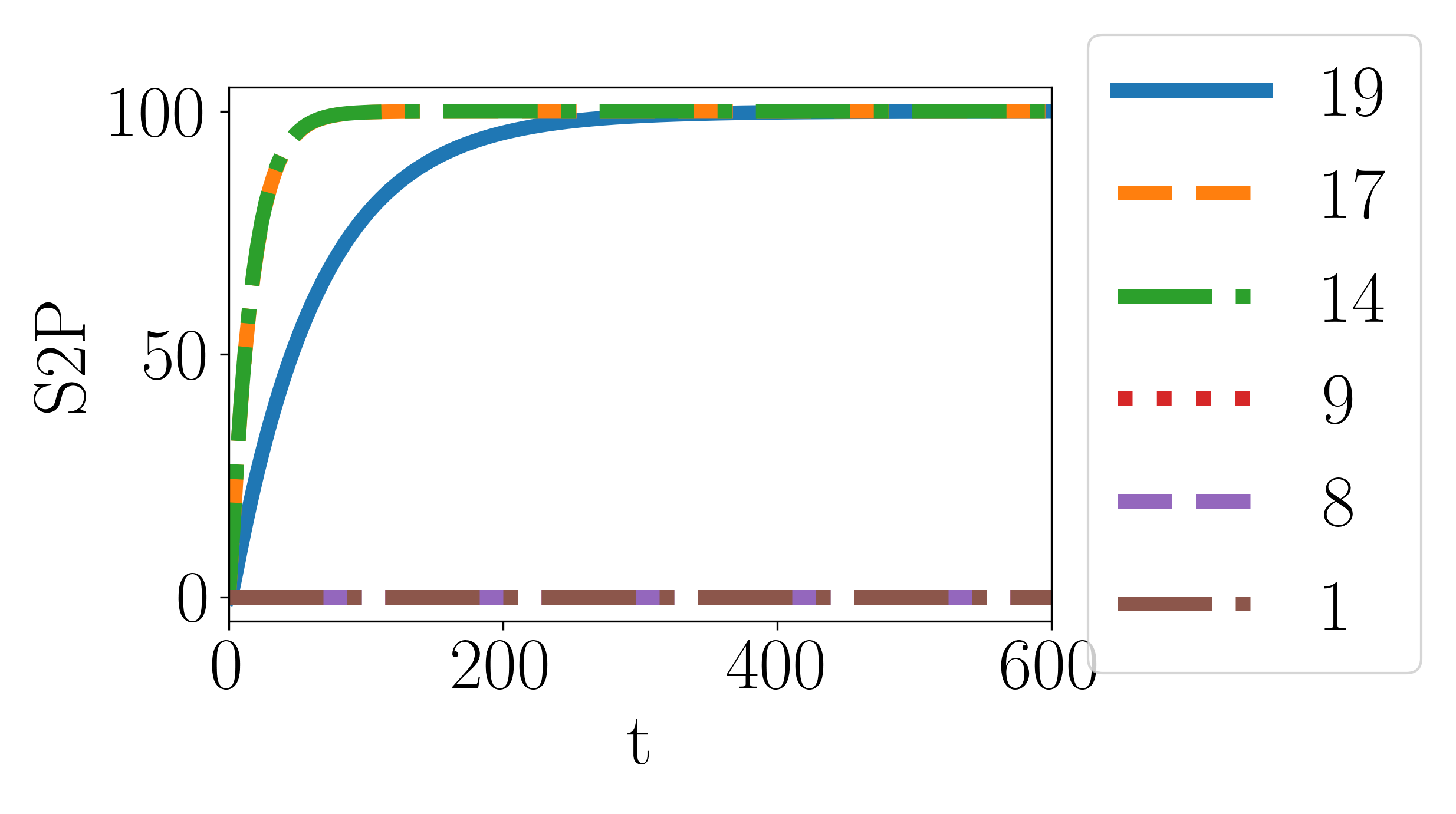}
                \vspace{-0.75cm}
        \caption{Experiment on Model 5}\label{fig:greaterredpower:5}
    \end{subfigure}\qquad\quad
    \begin{subfigure}[b]{0.45\textwidth}
        \includegraphics[width=\textwidth]{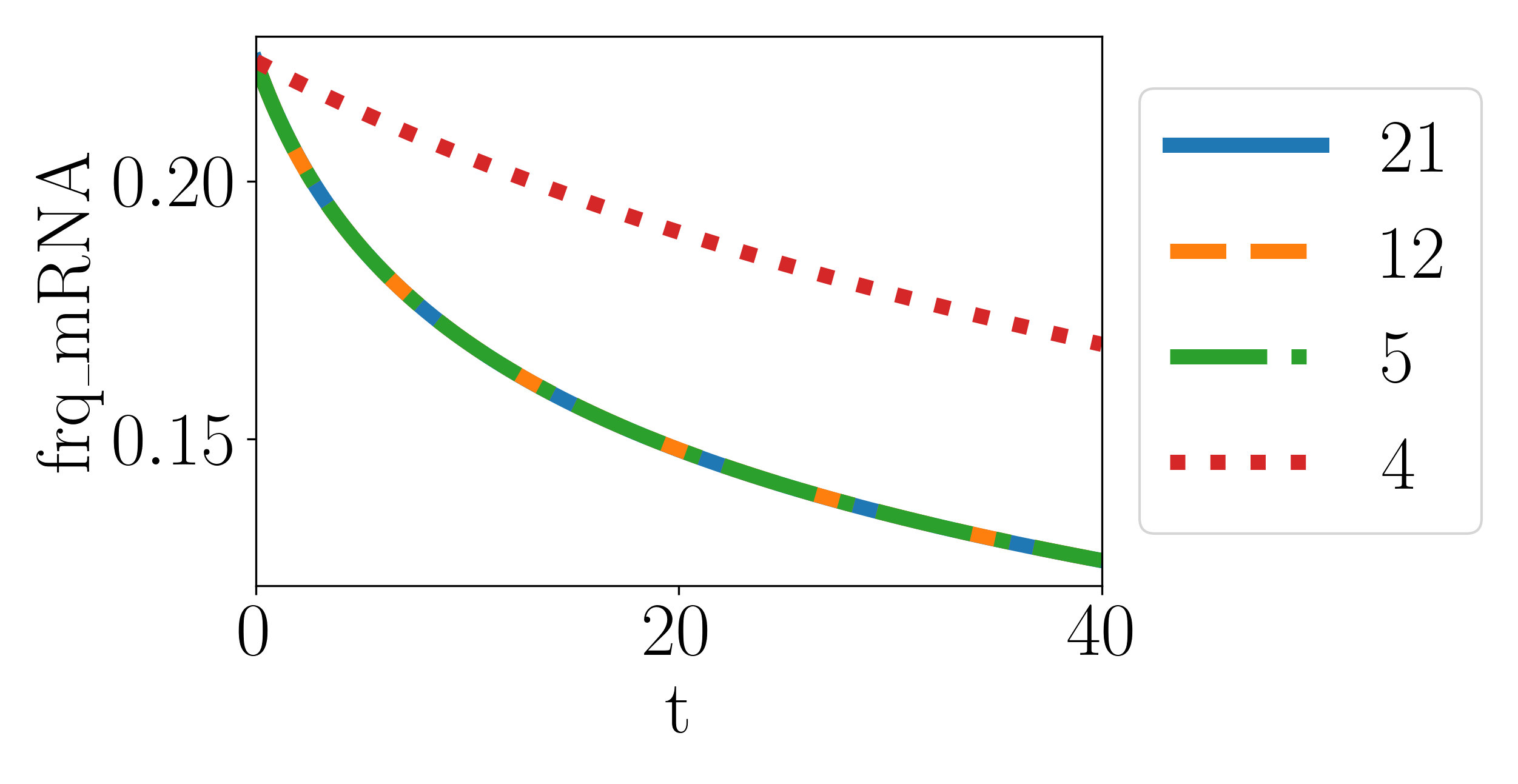}
                \vspace{-0.75cm}
        \caption{Experiment on Model 6}\label{fig:greaterredpower:6}
    \end{subfigure}\qquad\quad
    \begin{subfigure}[b]{0.45\textwidth}
        \includegraphics[width=\textwidth]{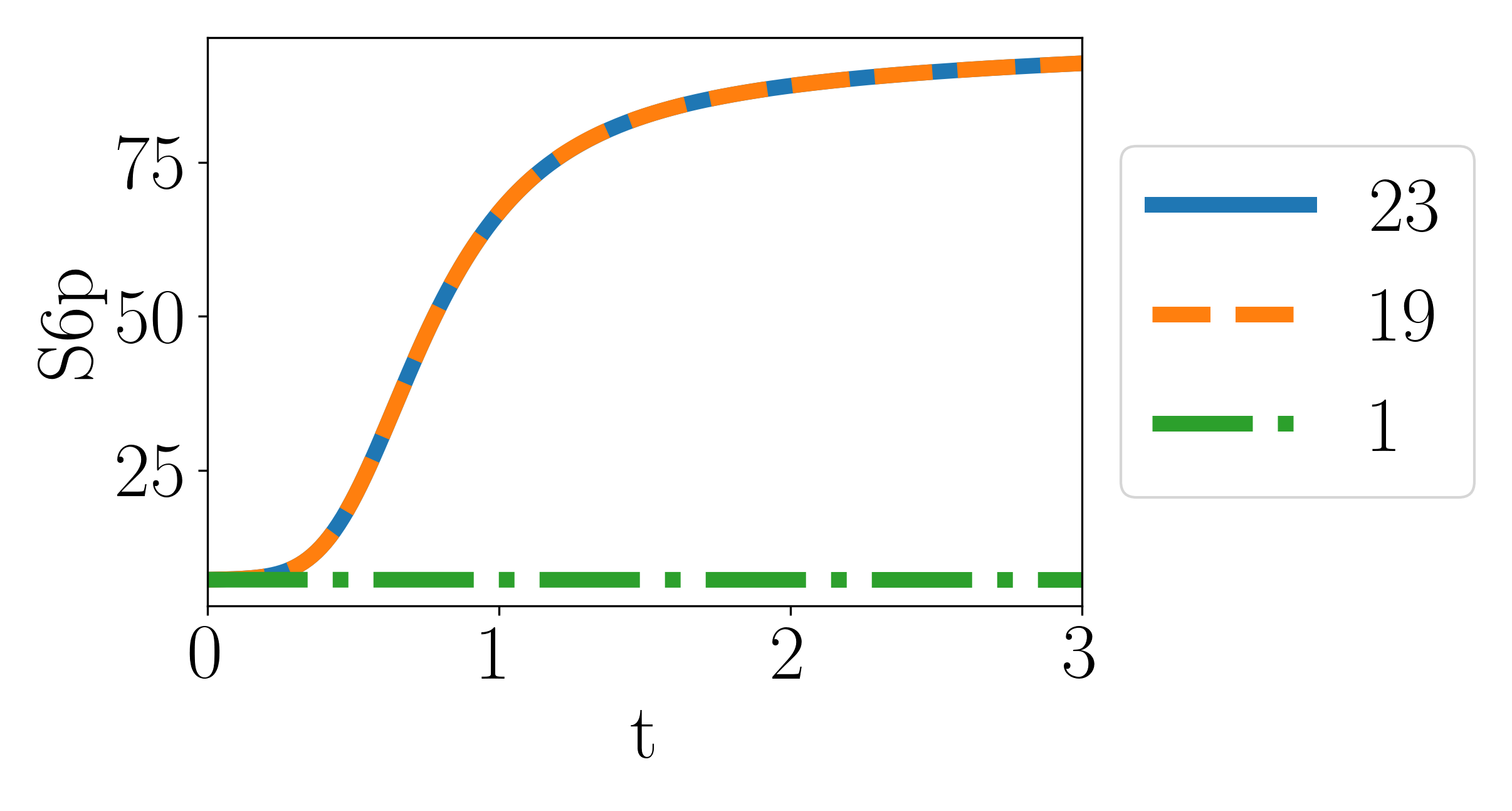}
                \vspace{-0.75cm}
        \caption{Experiment on Model 7}\label{fig:greaterredpower:7}
    \end{subfigure}\qquad\quad
    \begin{subfigure}[b]{0.45\textwidth}
        \includegraphics[width=\textwidth]{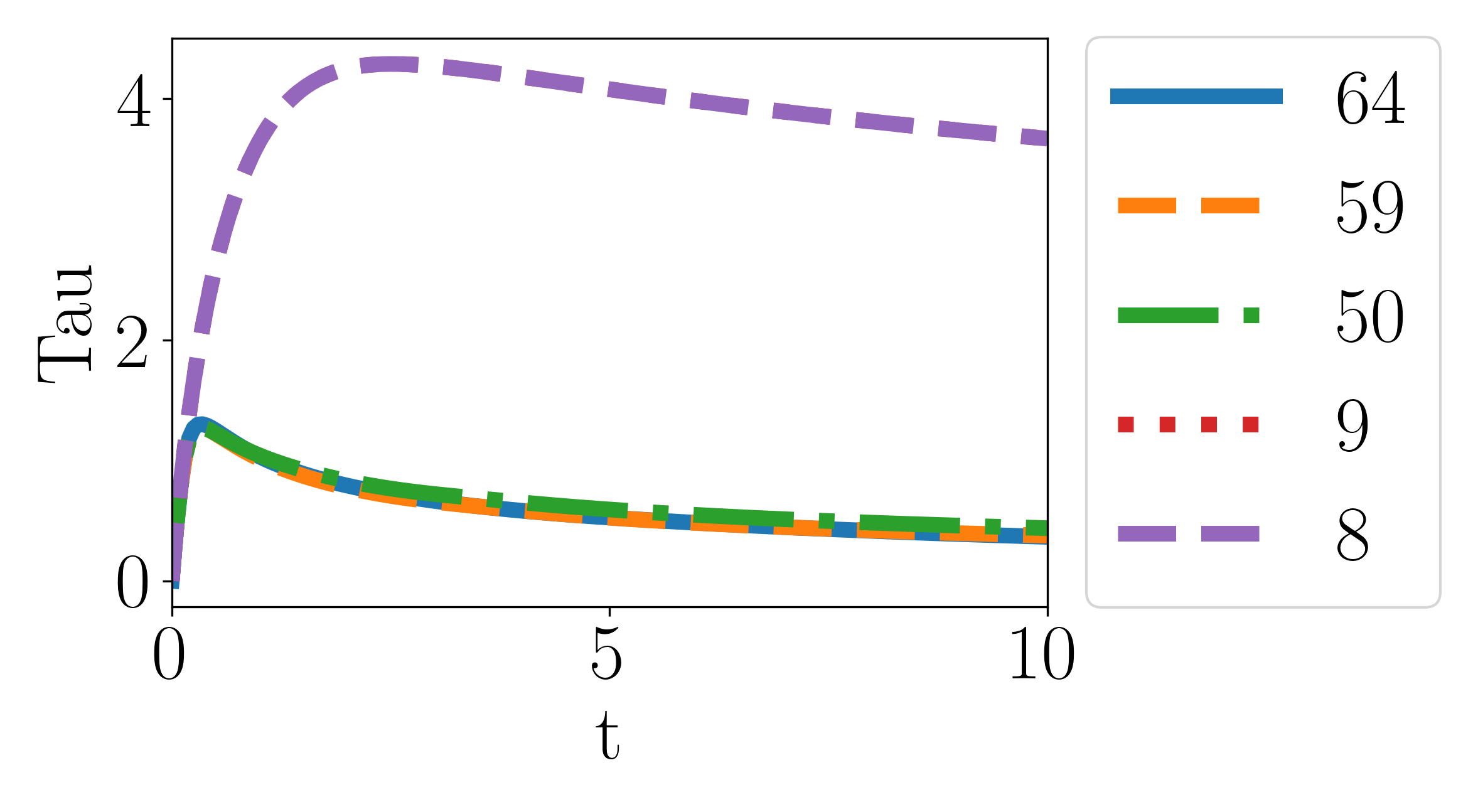}
                \vspace{-0.75cm}
        \caption{Experiment on Model 8}\label{fig:greaterredpower:8}
    \end{subfigure}
    \vspace{-0.3cm}
    \caption{
        Simulation of exact and approximately lumped models from Table~\ref{tab:greaterredpower}, for varying cutoff sizes.
        The continuous blue line refers to exact constrained lumping.
    }
    \label{fig:greaterredpower}
\end{figure}

\subsection{Scalability analysis}\label{sec_eval:scal}
We evaluate the scalability of our approach using a parametric model of multisite protein phosphorylation which can be made combinatorially larger by increasing its number of \emph{binding sites}~\cite{FEBS:FEBS7027}.
This mechanism is key in the study of cellular processes~\cite{Gunawardena11102005}. 
The model describes the phosphorylation and dephosphorylation of a substrate with $n$ different binding sites, each capable of being in 4 different states~\cite{sneddon2011efficient}.
The resulting model requires $4^n+2$ variables to track the evolution of all protein configurations as well as the kinase and phosphatase concentrations.

The exact constrained lumping of this model was studied in~\cite{ovchinnikov_clue_2021}.
We choose the concentration of free kinase as the observable of interest.
For our setup, we added random noise to each reaction, taken from  $-5\%$ to $+5\%$ of the parameter value appearing in each reaction (the same parameter could have different values for different reactions). 
This ensured that the perturbed model was not exactly reducible.
Next, we ran Algorithm~\ref{alg:acleps} with a maximum reduction size $m^*= 10$ and $d_{\min}=10^{-4}$ for models with 2 to 5 binding sites.
To compute the error, the time horizon was set to $2$, and the initial conditions were taken from the original paper.
All simulations were computed using the \texttt{LSODA} stiff solver.

\begin{table}[t]
\resizebox{\textwidth}{!}{
        \begin{tabular}{rrrrrrrrrr}
            \hline
            \multicolumn{1}{c}{\emph{Sites}} & \multicolumn{1}{c}{\emph{Size}} & \multicolumn{1}{c}{\emph{Red. Size}} & \multicolumn{1}{c}{\emph{Iter.}} & \multicolumn{1}{c}{\emph{Time (min)}} & \multicolumn{1}{c}{\emph{Original (s)}} &  \multicolumn{1}{c}{\emph{Red. (s)}} & \multicolumn{1}{c}{\emph{$e(T)$}} & \multicolumn{1}{c}{\emph{$e_{Rel}(T)$ (\%)}} & \multicolumn{1}{c}{\emph{$e_{\max}$}} \\
            \hline
            2 & 24 & 5 & 8 & 0.003 & 0.116 & 0.022 & 0.018 & 2.60 & 0.274 \\
            3 & 72 & 4 & 8 & 0.214 & 0.186 & 0.020 & 0.004 & 1.00 & 0.236 \\
            4 & 264 & 4 & 8 & 8.024 & 1.429 & 0.032 & 0.001 & 0.30 & 0.208 \\
            5 & 1032 & 4 & 8 & 371.683 & 13.303 & 0.025 & 0.001 & 0.60 & 0.199 \\
            \hline
        \end{tabular}
    }
    \caption{Approximate lumpings on multisite phosphorylation.}
    \label{tab:large_exp}
\end{table}

Table~\ref{tab:large_exp} presents the results for multisite phosphorylation models.
Column \emph{Sites} displays the number of binding sites in the model.
The size of the original model and that of the reduced model are shown in the columns \emph{Size} and \emph{Red. Size}, respectively.
The number of iterations and the total computation time of Algorithm~\ref{alg:acleps} are displayed in the columns \emph{Iter.} and \emph{Time}, respectively.
The columns \emph{Original} and \emph{Red.} contain the time used to compute the simulation using the non-reduced model and the reduced model, respectively.
The absolute and relative error at steady-state are shown in the columns $e(T)$ and $e_{Rel}(T)$.
While the maximum absolute error is shown in $e_{\max}$.

As seen in Table~\ref{tab:large_exp}, our approach can find a reduction with only 4 species for a model without exact reductions. 
These reductions have low errors (less than $3\%$ for all models).
This is a promising result for the study of more robust reductions.
The runtime depends on the size of the model. 
Larger models take considerably longer as there is a combinatorial increment in the number of monomials needed to represent the Jacobian. 
This, in turn, gets amplified by the number of iterations necessary to find a lumping with the given size restriction, which depends on the tolerance $d$ for Algorithm~\ref{alg:acleps}.

\section{Conclusion}\label{sec_conc}

We have proposed approximate constrained lumping, a relaxation of the exact reduction technique constrained lumping. This allows for more aggressive reductions at the cost of introducing, in a controlled way, errors in the dynamics of the reduced model.
The technique can be applied to dynamical systems with polynomial or rational derivatives, common in computational biology.
For a given dynamical system, our algorithm utilizes a numerical tolerance to efficiently compute a reduced system that preserves the evolution of a linear combination of specific variables of interest up to a bounded error.
We proved that the error bound is proportional to the numerical threshold.
Additionally, we introduced a heuristic to obtain appropriate numerical tolerances based on the model size. 
We demonstrated effectiveness and scalability by obtaining low-error reductions for a collection of published biological models.
Future work will explore robust reductions by considering models with uncertain kinetic parameters.

\bibliographystyle{plain}
\bibliography{cmsb}

\appendix

\subsection*{Appendix: Proofs}


\begin{proposition}\label{prop:errderbound}
	Let $\dot{x} = f(x)$ a system of ODEs such that $f$ is analytic on a compact set $\Omega\subseteq \mathbb{R}^{m}$.
	Given a $(S,T,\eta)-$lumping $L$, such that $x(t), \bar{L}Lx(t) \in \Omega$ for all $t \in [0,T]$, the norm of the error $ \norm{e(t)}$ is bounded as follows
	\begin{equation*}
		\norm{\dot{e}(t)} \leq \beta \norm{e(t)} + \eta, ~\forall t \in [0,T], ~\forall x(0) \in S. \label{eq:edotbound}
	\end{equation*}
\end{proposition}

\begin{proof}
	Note that for any initial condition $x(0)\in S$, we can compute
	\begin{align*}
		\norm{\dot{e}(t)} & =  \norm{\dot{y}(t)-L \dot{x}(t)}                                         \\
		                  & = \norm{ Lf(\bar{L}y) -Lf(x)}                                             \\
		                  & = \norm{ Lf(\bar{L}y) - Lf(\bar{L}Lx) + Lf(\bar{L}Lx) -Lf(x)}             \\
		                  & \leq \norm{ Lf(\bar{L}y) - Lf(\bar{L}Lx)}  + \norm{ Lf(\bar{L}Lx) -Lf(x)}
	\end{align*}
	
	To bound the first term, recall that $f$ is analytic on $\Omega$.
	Therefore $f$ is locally Lipschitz in $\Omega$.
	Then there is a constant $C$ such that
	\begin{equation*}
		\norm{ Lf(\bar{L}y) - Lf(\bar{L}Lx)} \leq C\norm{L} \norm{\bar{L}} \norm{e(t)},
	\end{equation*}
	for all $t \in [0,T]$ and all $x(0)\in S$.
	
	Using the fact that $L$ is an $(S,T,\eta)-$lumping, we get that
	\begin{equation*}
		\norm{\dot{e}(t)} \leq C\norm{L} \norm{\bar{L}}\norm{e(t)} + \eta,
	\end{equation*}
	for all $t \in [0,T]$ and all $x(0)\in S$. The result follows by setting $\beta = \norm{L}\norm{\bar{L}}$.
\end{proof}

\begin{lemma}\label{lem:groenwall}
	\cite[Lemma 2]{iacobelli_lumpability_2013}.
	Let $f$ be a continuous scalar function on $[0,T]$ which has a right derivative $D_{R}f(t)$ such that
	\begin{equation*}
		D_{R}f(t) \leq \beta f(t) + \gamma,
	\end{equation*}
	for all $t\in [0,T)$, where $\beta$ and $\gamma$ re constants and $f(0)=0$. Then
	\begin{equation*}
		f(t)\leq \frac{\gamma}{\beta}\left( e^{\beta t} -1 \right),
	\end{equation*}
	for all $t\in [0,T]$.
\end{lemma}

\begin{proof}[Theorem \ref{thm:errbound}]
	By hypothesis, there exists a compact set $\Omega \subseteq\mathbb{R}^{m}$ such that the $x(t), \bar{L}Lx \in \Omega$ for all $t \in [0,T]$ and for all $x(0) \in S$.
	Given that $f(x)$ is analytic in $\Omega$, it is possible to apply Lemma \ref{lem:groenwall} and Proposition \ref{prop:errderbound} to obtain the desired result.
\end{proof}


	\begin{proof}[Theorem~\ref{thm:complexity}]
        We first prove the polynomial complexity of Algorithm~\ref{alg:acl}.
		To this aim, we assume a naive implementation of the matrix operations.
		We begin by analyzing the complexity of the operations performed in the main loop (Line~\ref{alg:acl:mainloop}) of Algorithm~\ref{alg:acl}.
		To find the complexity of computing $\pi_{i} = rJ_{i}L^{T}L$, recall that $L$ is an $p\times m$ matrix.
		This means that computing $L^{T}L$ can be done in $\mathcal{O}(pm^2)$, while $rJ_{i}$ can be computed in $\mathcal{O}(m^2)$. Hence, $\pi_{i}$ can be computed in $\mathcal{O}((p+2)m^2)$. Since $rJ_i$ and $\pi_i$ are of dimensions $1\times m$, computing $d_i = rJ_i - \pi_i$ has complexity $\mathcal{O}(m)$. Instead, computing $d_i/\norm{d_i}$ has complexity $\mathcal{O}(4m)$. Overall, we have that the complexity of the operations inside the main loop (Line~\ref{alg:acl:mainloop}) of Algorithm~\ref{alg:acl} is $\mathcal{O}((p+2)m^2)$. As this loop is computed $Np$ times, we have that the total complexity of Algorithm~\ref{alg:acl} is $\mathcal{O}(Np(p+2)m^2)$. The result follows from the fact that Algorithm~\ref{alg:findJ} can be implemented with a complexity of $\mathcal{O}(m^2A)$, where $A$ is the arithmetic cost of computing an entry of $f(x)$.
	\end{proof}



\begin{theorem}\label{thm:schs}
	\cite[Theorem 9.19]{rudin_principles_1976}.
	Let $\Omega \in \mathbb{R}^{m}$ open and convex and let $f \in C^{1}(\Omega,\mathbb{R}^{n})$.
	If there is a constant $C<\infty$ such that
	\begin{equation*}
		\norm{ Df(x)}  \leq C ~\forall x \in \Omega
	\end{equation*}
	then
	\begin{equation*}
		\norm{ f(x_{1}) - f(x_{0})} \leq C  \norm{x_{0} - x_{1}}
	\end{equation*}
\end{theorem}

\begin{proposition} \label{prop:bound}
	Let $L$ be the matrix computed via Algorithm \ref{alg:acl} with tolerance $\varepsilon$ and let $\Omega \subset \mathbb{R}^{m}$ be an open, bounded and convex set.
	If $x \in \Omega$ and $\bar{L}Lx \in  \Omega $, $f$ is analytic in $\bar{\Omega}$, and $J(x)$ can be written in the form given by Equation~\ref{eq:jacrep}, then it follows that
	\begin{equation}
		\norm{ L f(\bar{L}L x) - L f(x)} \leq C\varepsilon \norm{ \bar{L}L x - x},
	\end{equation}
	where $C$ is a constant.
\end{proposition}

\begin{proof}
	Fix $x \in \Omega$, by hypothesis $x^{R} := \bar{L}Lx \in \Omega $.
	Let $g:\ker L \to \mathbb{R}^{m}, ~v\mapsto f(v +x^{R}) $.
	
	
	We claim that
	\begin{equation}
		\norm{DLg(v)} \leq C \varepsilon,~ \forall v \in \pi(\Omega),
	\end{equation}
	where $C$ is a constant, and $\pi$ is the projection onto $\ker L$.
	
	Note that
		$DLg(v) = LJ(v + x^{R}))\iota$,
	where $J$ is the Jacobian of $f$ and $\iota: \mathbb{R}^{l} \to \mathbb{R}^{m}$ is the inclusion given by $v\mapsto (v,0,\dots, 0)$.
	
	We compute
	\begin{equation} \label{eq:norm}
		\norm{DLg(v)} = \norm{LJ(v+x^{R})\iota} = \norm{LJ(v+x^{R})}^{\ker L },
	\end{equation} \label{eq:normdef}
	where $\norm{\cdot}^{\ker L}$ is the norm restricted to $\ker L$.
	
	Following Equation~\eqref{eq:norm}, we have to estimate
	\begin{equation}
		\norm{LJ(v+ x^{R})u}, \forall v \in \pi(\Omega), u \in \ker L \text{ s.t. } \norm{u} =1 .
	\end{equation}
	
	To this aim, we will first estimate each entry of the vector $LJ(v+x^{R})u$.
	Denote by $e_{i} \in \mathbb{R}^{m}$, the vector of zeroes with one in the $i-$th entry, and denote by $r_{i}$ the $i-$th row of $L$ .
	Let $p = v+ x^{R}$.
	Given  that $u \in \ker L$, we have that
	\begin{align*}
		\norm{e_{j}LJ(p)u} & = \norm{r_{j}J(p)u}                                  
		                    = \norm{r_{j} \left( J(p) - J(p)\bar{L}L \right) u}.
	\end{align*}
	Using the fact that $J(x) = \sum_{i = 1}^k  J_i \mu_i(x)$, we have
	\begin{align*}
		\norm{e_{j}LJ(p)u} & = \norm{r_{j}J(v+ \bar{L}Lx)u}                                                                    \\
		                   & = \norm{r_{j} \left( \sum_{i = 1}^k  \left( J_i \mu_i(p) - J_i \mu_i(p)\bar{L}L\right) \right) u} = \norm{ \sum_{i = 1}^k \mu_i(p) \left( r_{j}J_i  - r_{j}J_i \bar{L}L\right)u  }.
	\end{align*}
	Using the triangle inequality, and the Cauchy-Schwartz inequality we have that
	\begin{align}\label{eq:boundineq}
		\nonumber
		\norm{e_{j}LJ(p)u} & \leq \sum_{i = 1}^k\norm{  \mu_i(p) \left( r_{j}J_i  - r_{j}J_i \bar{L}L\right)u  }                    
		                   \leq \sum_{i = 1}^k\left| \mu_i(p)\right| \norm{ \left( r_{j}J_i  - r_{j}J_i \bar{L}L\right)} \norm{u} \\
		                   & \leq \sum_{i = 1}^k\left| \mu_i(p)\right|\norm{ \left( r_{j}J_i  - r_{j}J_i \bar{L}L\right)}           
		                   \leq \sum_{i = 1}^k\left| \mu_i(p)\right|  \varepsilon,
	\end{align}
	where $\varepsilon$ is the tolerance used in the computation of Algorithm \ref{alg:acl}.
	
	Since $f$ is analytic in $\bar{\Omega}$, we have that its derivative is bounded in $\Omega$.
	Therefore each $\mu_{i}$ is bounded in $\Omega$ and so, by the extreme value theorem,
	there exists a constant $C$  such that
	\begin{equation}\label{eq:boundmu}
		\sup_{i,v \in \pi(\Omega)} \left| \mu_{i}(v+x^{R})\right| = C.
	\end{equation}
	Combining Equation \eqref{eq:boundmu} with the previous reasoning we have that
	\begin{equation}\label{eq:entrybnd}
		\norm{e_{j}LJ(p)u} \leq C\varepsilon.
	\end{equation}
	Using Equations \eqref{eq:norm} and \eqref{eq:entrybnd}, we get that
	\begin{equation}\label{eq:bounddg}
		\norm{DLg(v)} \leq \sqrt{m} C \varepsilon, \forall v \in \pi(\Omega).
	\end{equation}
	
	As $\pi$ is a linear map, the set $\pi(\Omega)$ is open and convex.
	Therefore, we can use Equation \eqref{eq:bounddg} and Theorem \ref{thm:schs} to show that
	\begin{equation}\label{eq:boundg}
		\norm{g(v_{0})-g(v_{1})}\leq \sqrt{m}C\varepsilon \norm{v_{0} - v_{1}}, \forall v_{0},v_{1} \in \pi(\Omega).
	\end{equation}
	
	Notice that  $x^{R} = \bar{L}Lx$ is the projection of $x$ onto the rowspace of $L$.
	It follows that $\pi(\bar{L}Lx)=0$, and so $0 \in \pi(\Omega)$.
	Moreover since $x\in \Omega$, it follows that $\pi(x)= x -\bar{L}Lx\in \pi(\Omega)$.
	So we can set $v_{0}= 0$ and $v_{1} = \pi(x)$ in Equation \eqref{eq:boundg} to get  that
	\begin{align*}
		\norm{g(0)-g(\pi(x))}    & \leq \sqrt{m}C\varepsilon \norm{0 - \pi(x)}     \\
		\norm{f(\bar{L}Lx)-f(x)} & \leq \sqrt{m}C\varepsilon \norm{  \bar{L}Lx-x},
	\end{align*}
	where there last result follows from the definition of $x^{R}$ and the fact that $x = x^{R} + \pi(x)$, as we wanted to prove.
\end{proof}

\begin{proposition}\label{prop:analytic}
	Let $f: \RE^m \rightarrow \RE^m$ be an analytic function on a compact set $\Omega \subset \RE^m$.
	Assume we can write $J(x)$ in the form given by Equation~\ref{eq:jacrep}, where $J_i \in \RE^{m\times m}$ and $\{\mu_i(x)\ :\ i=1,\ldots,N\}$ are $\RE$-linearly independent and analytic on $\Omega$.
	Consider $x_1,\ldots,x_N \in \Omega$ such that  $\{J(x_1),\ldots,J(x_N)\}$ is a basis of the vector space $\mathcal{V}_{J}$.
	Given a matrix $L\in \RE^{p\times m}$ if there is an $\varepsilon > 0$ such that for all $i=1,\ldots,N$ and all rows of $r$ of $LJ(x_i)$
	\[||r - \bar{L}Lr||_2 < \varepsilon,\]
	then there exists a constant $C>0$ such that
	for all $i=0,\ldots,N$ and all rows $s$ of $LJ_i$ we have
	\[||s - \bar{L}Ls||_2 < C\varepsilon.\]
\end{proposition}

\begin{proof}	
	It is clear that the matrices $J_i$ generate the space $\mathcal{V}_{J}$ and they belong to it.
	At the same time, we have by assumption that
	$J(x_i)$ is a basis of $\mathcal{V}$.
	Hence:
	\[J_i = c_{i,1}J(x_1) + \ldots + c_{i,N}J(x_N).\]
	Multiplying by $L$ to the left, we obtain:
	$LJ_i = c_{i,1}LJ(x_1) + \ldots + c_{i,N}LJ(x_N)$.
	
	Therefore, all rows of $LJ_i$ are linear combinations of rows of the matrices $LJ(x_1),\ldots,LJ(x_N)$.
	Let us denote by $s_{i,j}$ and $r_{i,j}$ the $j$-th row of $LJ_i$ and of $LJ(x_i)$, resp. Then we have:
	$s_{i,j} = c_{i,1}r_{1,j} + \ldots + c_{i,N}r_{N,j}$, which means that:
	\begin{align*}
		||s_{i,j} - s_{i,j}\bar{L}L||_2 & = \left|\left|c_{i,1}\left(r_{1,j} - r_{1,j}\bar{L}L\right) + \ldots + c_{i,N}\left(r_{N,j} - r_{N,j}\bar{L}L\right)\right|\right|_2 \\
		                                & \leq
		\left|\left|c_{i,1}\left(r_{1,j} - r_{1,j}\bar{L}L\right)\right|\right|_2 +
		\ldots + \left|\left|c_{i,N}\left(r_{N,j} - r_{N,j}\bar{L}L\right)\right|\right|_2                                                                                     \\
		                                & \leq N\max\{|c_{i,1}|,\ldots,|c_{i,N}|\}\varepsilon
	\end{align*}
	
	Thus, the statement is true by taking $C = N\max\{|c_{i,j}|\ :\ i,j = 1,\ldots,N\}$.
\end{proof}

\begin{lemma}\label{lmm:boundJ}
	In the context of Proposition~\ref{prop:bound}, if there is a set of $\mathbb{R}^{m\times m}$ matrices $\{J_1,\dots, J_N\}$ s.t. they span $\mathcal{V}_{J}$, then
		$\norm{f(\bar{L}Lx)-f(x)}\leq \sqrt{m}C'C\varepsilon \norm{  \bar{L}Lx-x},$
	where $C$ is the constant given by Proposition~\ref{prop:analytic} and $C' =\sup_{i, x\in \Omega} \left\lvert \mu_{i}(x) \right\lvert$
\end{lemma}

\begin{proof}
	Recall, the extended Algorithm~\ref{alg:acl} checks that
	$ \norm{rJ(x_{i}) - \bar{L}L rJ(x_{i})} \leq \varepsilon, $
	for all rows $r$ of $L$ and all $i = 1, \dots, N$, where $x_1,\ldots,x_N \in \Omega$.
	By using Proposition~\ref{prop:analytic}, we get that there is a constant $C$ such that
	$		\norm{rJ{i} - \bar{L}L rJ_{i}} \leq  C \varepsilon$
	for all $i = 1,\dots, N$, where, by hypothesis, we can write $J(x) = \sum_{i=0 }^{N} J_{i} \mu_{i}(x)$  with $J_i \in \RE^{n\times n}$ and $\{\mu_i(x)\ :\ i=1,\ldots,N\}$ being $\RE$-linearly independent and analytic on $\Omega$.	We follow the reasoning of the proof of Proposition~\ref{prop:bound}, up to Equation~\ref{eq:boundineq}.
	There by the fact that each $\mu_{i}(x)$ is bounded, we have that
	\begin{equation*}
		\norm{e_{j}LJ(p)u}\leq \sum_{i = 1}^k\left| \mu_i(p)\right|\norm{ \left( r_{j}J_i  - r_{j}J_i \bar{L}L\right)} \leq C'C \epsilon,
	\end{equation*}
	where $C'=\sup_{i,v \in \pi(\Omega)} \left| \mu_{i}(v+x^{R})\right| $.
	The result follows by continuing the aforementioned reasoning.
\end{proof}

\begin{proof}[Theorem~\ref{thm:algorithm:correctness}] As $f$ is analytic in $\bar{\Omega}$, we can use Lemma \ref{lmm:boundJ} for each $t \in [0,T]$ and each solution $x(t)$ with initial conditions in $S$ to get
	\begin{equation} \label{eq:thm4:1}
		\norm{ L f(\bar{L}L x(t)) - L f(x)} \leq \sqrt{m}C'C\varepsilon \norm{ \bar{L}L x(t) - x(t) },
	\end{equation}
	for all $t\in [0,T]$.
	
	Since $\Omega$ is compact by hypothesis, it is closed and bounded.
	Therefore, there exists an open ball $B(K/2)$ with center in $\Omega$ such that $\Omega \in B(K/2)$.
	Consequently, we obtain the following bound
	\begin{equation}\label{eq:thm4:2}
		\norm{ \bar{L}L x(t) - x(t) } <  K.
	\end{equation}
	The result follows by combining Equations~\eqref{eq:thm4:1} and \eqref{eq:thm4:2}
	
\end{proof}


\begin{proof}[Lemma \ref{lmm:uppereps}]
	Consider a  matrix of observables $M\in \mathbb{R}^{l\times m}$ of rank $l$ with $j-$th row denoted by $r_{j}$ and a decomposition of the Jacobian $J(x) = \sum_{i=1}^{n}J_{i}\mu_{i}(x)$.
	Let $\varepsilon_{max}$ be given by Equation \ref{eq:epsmax}.
	It follows that
	$\norm{r_{j} J_i - r_{i}J_i P_L} \leq \varepsilon_{max},$
	for all $j = 1, \dots, l$ and  $i=1,\dots,n$.
    Using $\varepsilon_{max}$ (or any $\varepsilon> \varepsilon_{max}$ as input for Algorithm \ref{alg:acl} implies that the condition in Line \ref{alg:acl:check} will be false.
	Thus, exiting the loop and returning a matrix whose rows are orthonormal rows spanning  $\rsp(M)$.

    To show that $\varepsilon_{max}$ is the lowest upper bound on $\varepsilon$, assume $\varepsilon^*\leq\varepsilon_{max}$ is such that 
    Algorithm~\ref{alg:acl} outputs orthonormal rows spanning $\rsp(M)$.
    This implies that 
    $\max_{i,j}\norm{r_{j} J_i - r_{i}J_i P_L} \leq \varepsilon^*$.
    By the definition of $\varepsilon_{max}$, it follows that $\varepsilon^* = \varepsilon_{max}$, as we wanted to prove.
\end{proof}

\end{document}